\documentclass[12pt,letterpaper]{scrartcl}
\usepackage[utf8]{inputenc}
\usepackage{tikz}
\usepackage{amsthm,amsfonts, mathtools, amssymb, bbm,amsmath, mathdots}
\usetikzlibrary{matrix}
\usetikzlibrary{cd}
\usepackage[margin=1in]{geometry} 
\usepackage[square,numbers,sort&compress]{natbib}
\usepackage{graphicx}
\usepackage{transparent}
\usepackage[percent]{overpic}
\usepackage{caption}
\usepackage{subcaption}
\usepackage{xcolor,color,soul} 
\usepackage{enumitem} 
\usepackage{physics}
\usepackage{qtree}
\usepackage[export]{adjustbox} 
\usepackage{appendix}
\definecolor{Myblue}{rgb}{0,0,0.6}  
\usepackage[colorlinks,citecolor=Myblue,linkcolor=Myblue,urlcolor=Myblue,pdfpagemode=None]{hyperref}
\allowdisplaybreaks 

\newcommand{\boxpic}[3]{
	\begin{tikzpicture}[baseline={([yshift=-.5ex]current bounding box.center)}]
	\node at (0,0) {\begin{overpic}[scale=#1]{#2}#3\end{overpic}};
	\end{tikzpicture}
}

\theoremstyle{definition}
\newtheorem{definition}{Definition}
\newtheorem{theorem}[definition]{Theorem}
\newtheorem{lemma}[definition]{Lemma}
\newtheorem{proposition}[definition]{Proposition}
\newtheorem{example}[definition]{Example}

\newtheorem{remark}[definition]{Remark}

\numberwithin{definition}{section}
\numberwithin{equation}{section}
\numberwithin{figure}{section}

\newcommand{\be}{\begin{equation}}
\newcommand{\ee}{\end{equation}}

\newcommand{\id}[1]{\operatorname{id}_{#1}} 
\newcommand{\unit}{\mathbbm{1}} 
\newcommand{\catname}[1]{{\mathcal{#1}}}
\newcommand{\pairing}{\mathrm{d}}
\newcommand{\doublecat}{\mathcal{C}\boxtimes\mathcal{C}^\mathrm{rev}}
\newcommand{\Vect}{\operatorname{Vect}}
\newcommand{\SL}{\operatorname{SL}(2,\mathbb{Z})}

\newcommand{\AdS}{\mathrm{AdS}_3}

\newcommand{\Hom}{\mathrm{Hom}} 
\newcommand{\End}{\mathrm{End}} 
\newcommand{\Aut}{\mathrm{Aut}}
\newcommand{\Rep}{\mathrm{Rep}}
\newcommand{\Mod}{\mathrm{Mod}}
\newcommand{\PMod}{\mathrm{PMod}}

\newcommand{\colG}{\mathrm{col}(\Gamma)}
\newcommand{\colGo}{\mathrm{col}(\Gamma)^{\circ}}
\newcommand{\tcolGo}{\widetilde{\mathrm{col}}(\Gamma)^{\circ}}

\newcommand{\ie}{i.e.\,}
\newcommand{\eg}{e.g.\,}

\begin{document}

\title{CFT correlators and\\
mapping class group averages}

\author{
	Iordanis Romaidis$^\ast$ \qquad
	Ingo Runkel$^\vee$\\[0.5cm]
	\normalsize{\texttt{\href{mailto:iromaidi@ed.ac.uk}{iromaidi@ed.ac.uk}}} \\  
	\normalsize{\texttt{\href{mailto:ingo.runkel@uni-hamburg.de}{ingo.runkel@uni-hamburg.de}}}\\[0.5cm]
        \normalsize\slshape $^\ast$School of Mathematics, University of Edinburgh,\\
	\normalsize\slshape Mayfield Road Edinburgh EH9 3FD, UK\\
	\normalsize\slshape $^\vee$Fachbereich Mathematik, Universit\"{a}t Hamburg,\\
	\normalsize\slshape Bundesstra{\ss}e 55, 20146 Hamburg, Germany
	}

\date{}

\maketitle

\begin{abstract}
Mapping class group averages appear in the study of 3D\,gravity partition functions. In this paper, we work with 3D\,topological field theories
to establish a bulk-boundary correspondence between such averages and correlators of 2D\,rational CFTs whose chiral mapping class group representations are irreducible and satisfy a finiteness property. 
We show that Ising-type modular fusion categories satisfy these properties on surfaces with or without field insertions, extending results in \cite{Jian:2019ubz}, and we comment on the absence of invertible global symmetries in the examples we consider.
\end{abstract}

\newpage

\setcounter{tocdepth}{2}
\tableofcontents

\section{Introduction}\label{sec:intro}

Mapping class group averages appear in 3-dimensional quantum gravity when trying to give a computable approximation of the gravity partition function. This surprising aspect of mapping class group actions offers a natural gateway to study the relation between a candidate 3D\,quantum gravity theory in the bulk and a 2D\,CFT on the boundary in the spirit of the AdS/CFT correspondence. The goal of this paper is to establish a correspondence between mapping class group averages and certain rational CFTs using the associated 3D\,Reshetikhin-Turaev TQFT. 

Consider the gravity partition function in pure Einstein gravity as the path integral
\begin{equation}\label{eq:intro-MCG}
    Z_\text{grav}(\Sigma) 
    ~=  \hspace{-.5em} \sum_{M\text{ topologies}}\int{\mathcal{D}g \, e^{iS[g]}}
    \ ,
\end{equation}
where the sum runs over diffeomorphism classes of smooth 3-manifolds $M$ with a fixed conformal boundary $\Sigma$ and the path integral is over Riemannian metrics $g$ on $M$. Note that the path integral is ill-defined but it serves as motivation to find meaning in summing over 3-manifolds.

A family of such 3-manifolds is obtained by mapping classes of the conformal boundary $\Sigma$. Namely, by twisting the boundary of the handlebody $H_\Sigma$ by a mapping class $\gamma$ one obtains a new 3-manifold $H_\Sigma^\gamma$ which bounds $\Sigma$. The result is a $\Mod(\Sigma)$-family of 3-manifolds where $\Mod(\Sigma)$ denotes the mapping class group. Restricting 
the sum in \eqref{eq:intro-MCG}
to the contribution of this $\Mod(\Sigma)$-family in the partition function gives a mapping class group average. In fact, in the case where $\Sigma$ is the conformal torus, in~\cite{Maloney:2007ud} it is argued 
that in the semi-classical limit\footnote{Recall that the Brown-Henneaux central charge is given by
$c = \frac{3l}{2G}$,
where $l$ is the AdS radius and $G$ is Newton's constant \cite{BH}. The semi-classical limit refers to the limit $c\rightarrow \infty$.}  
this contribution dominates the whole partition function. 
The $\Mod(\Sigma)$-family is typically still infinite and there is no clear way to average over it. 
We circumvent this issue by working with theories that obey a finiteness property that allows us to write mapping class group averages as finite sums.

Torus mapping class group averages for unitary Virasoro minimal model CFTs are computed in~\cite{Castro:2011zq}. 
Recall that the mapping class group of the torus $T^2$ is $\Mod(T^2) = \SL$. 
The solid torus with twisted boundary $H_{T^2}^\gamma$ is interpreted as a euclidean black hole \cite{BTZ} 
and hence the torus mapping class group average computes the contribution of the $\SL$-family of black holes.
Under the assumption that the vacuum contribution is $Z_\text{vac}(\tau) = |\chi_{0}(\tau)|^2$, where $\chi_{0}(\tau)$ denotes the holomorphic vacuum character, the contribution of the $H_{T^2}^\gamma$ is $Z_\gamma(\tau) = Z_\text{vac}(\gamma. \tau)$, where $\gamma.\tau$ denotes the usual M\"obius transformation of $\tau$ by $\gamma$. After regularising the sum by identifying $H_{T^2}^\gamma\sim H_{T^2}^{\gamma'}$
whenever $Z_\gamma = Z_{\gamma'}$, the mapping class group average of $Z_\text{vac}$ is 
\begin{equation}\label{eq:intro-torus-mcg-average}
\langle Z_\text{vac} \rangle_{T^2} ~\propto \sum_{Z_{\gamma} \in \mathcal{O}_{\text{vac}}}{Z_\gamma}\ ,
\end{equation}
where $\mathcal{O}_\text{vac}$ denotes the mapping class group orbit of $Z_{\text{vac}}$. 
The crucial point is that for Virasoro minimal models, and in fact for all rational CFTs (RCFTs) \cite{Ng}, this sum is now finite.
Note that this does in general not hold for higher genus surfaces $\Sigma$, except for when the RCFT has the so-called ``property F'' with respect to $\Sigma$, meaning that the representation image of $\Mod(\Sigma)$ is finite. 

\medskip

In this paper we work with the TQFT approach to RCFT developed in \cite{FRS1,FRS4,FRS5,Fjelstad:2006aw}.
Given a rational VOA $\mathcal{A}$, its representations form a modular fusion category (MFC) $\mathcal{C} = \Rep(\mathcal{A})$ \cite{Huang}. This is the input for the Reshetikhin-Turaev
3D\,TQFT. 
For a surface $\Sigma$ possibly with framed points, its state space 
$V^\mathcal{C}(\Sigma)$
is naturally a (projective) mapping class group representation. This state space corresponds to the space of chiral conformal blocks on $\Sigma$ and the full conformal blocks are the elements of 
$V^\mathcal{C}(\widehat\Sigma)$
where $\widehat\Sigma$ denotes the double of $\Sigma$. Conformal correlators $\operatorname{Cor}^{\catname C}_A(\Sigma)$ are elements of $V^\mathcal{C}(\widehat\Sigma)$ which are mapping class group invariant and satisfy a factorisation (or sewing) property. The latter refers to a compatibility of correlators under cutting and gluing, but only mapping class group invariance will play a role in our considerations. 
The index $A$ labels different consistent collections of correlators and refers to a special symmetric Frobenius algebra in $\catname C$.

The mapping class group average, if well-defined, is by definition mapping class group invariant and hence itself a candidate for a CFT partition function.
For unitary Virasoro minimal models, 
the sum \eqref{eq:intro-torus-mcg-average} is proportional to a CFT partition function on the torus precisely for the Ising and the tricritical Ising CFT \cite{Castro:2011zq}. 
For other central charges, it produces a sum which includes independent CFT partition functions and even unphysical modular invariants.
Moreover, only the Ising case extends to higher genus surfaces. 
In that case we have, for a surface $\Sigma$ without framed points,
\begin{equation}\label{eq:intro-Jian}
    \langle x_\text{vac} \rangle_{\Sigma} ~\propto~ Z_\text{CFT}(\Sigma) \ ,
\end{equation}
where $x_\text{vac}$ is the higher genus analogue of the vacuum contribution and $Z_\text{CFT}(\Sigma)$ is the RCFT partition function of the unitary Ising CFT 
\cite{Jian:2019ubz}. 
It is furthermore shown in \cite{Jian:2019ubz} that unitary Ising MFCs satisfy property F and have irreducible mapping class group representations with respect to surface with no framed points.   

We generalise these results in two directions. Firstly, \eqref{eq:intro-Jian} applies to more general RCFTs and to surfaces with framed points. This is captured by following bulk-boundary correspondence theorem (see Theorem~\ref{thm:mcgavg-cor}): 

\begin{theorem}\label{thm:intro-mcgavg-cor}
Let $\Sigma$ be a surface with framed points such that the projective representation $V^{\catname C}(\Sigma)$ is irreducible and the representation image $V^{\catname C}(\Mod(\Sigma)) \subset \End(V^{\catname C}(\Sigma))$ is finite.
If $x\in V^{\catname C}(\widehat\Sigma)$ satisfies a non-degeneracy condition\footnote{There is a natural functional $\pairing_{\Sigma}: V^\mathcal{C}(\widehat\Sigma) \rightarrow \mathbbm{k}$ described in Section~\ref{sec:RT} and we require that $\pairing_\Sigma(x) \neq 0$.}, 
then the average $\langle x \rangle_\Sigma$ is non-zero and
for each choice of $A$
there exists $\lambda_\Sigma\in \mathbbm{k}$ such that 
\begin{equation}\label{eq:intro-mcgavg-cor}
	 \operatorname{Cor}^{\catname C}_A(\Sigma) = \lambda_\Sigma\, \langle x \rangle_\Sigma~.
\end{equation}
\end{theorem}

Our second generalisation is to show property F and irreducibility for all MFCs with Ising fusion rules and for all surface with framed points (see Theorems \ref{thm:Ising-irred} and \ref{thm:Ising-propF}):

\begin{theorem}\label{thm:intro-ising-irred-propF}
    Let $\mathcal{C}$ be an MFC of Ising-type and let $\Sigma$ be a surface whose framed points are labelled by simple objects. Then $V^\mathcal{C}(\Sigma)$ is irreducible and the representation image $V^\mathcal{C}(\Mod(\Sigma))$ is finite. 
\end{theorem}

Irreducibility and property F which are part of the hypothesis in Theorem~\ref{thm:intro-mcgavg-cor} are of independent mathematical interest. In fact, we have studied in \cite{RR} MFCs $\mathcal{C}$ with irreducible mapping class group representations and we have found that they possess a unique Morita class of simple non-degenerate Frobenius algebras. 
This turns out to be connected to the conjecture on the absence of global symmetries in quantum gravity \cite{Harlow:2018jwu} 
as we explain in Section~\ref{subsec:global-sym}.
Property F is conjectured to hold for all weakly-integral MFCs \cite{NR} (at least with respect to all braid group representations) meaning for MFCs whose simple objects $i\in I$ have Frobenius-Perron dimensions which square to integers.

Not imposing irreducibility can allow for a higher-dimensional space of mapping class group invariants. Hence, mapping class group averages can produce a superposition of independent CFT partition functions, 
but also unphysical modular invariants,
depending on their input. 
This is related to the idea that there is no single CFT on the boundary, but rather the boundary theory is an ensemble average over a family CFTs \cite{CotlerJensen, CollierM}. 

\medskip

Let us describe the outline of the paper. In Section~\ref{sec:RT+MCG} we recall the mapping class group representations given by the Reshetikhin-Turaev TQFT and provide the action in terms of the mapping class group generators explicitly. Section~\ref{sec:conformal-corr} begins with a brief overview of the topological approach to RCFT, defines mapping class group averages and establishes the bulk-boundary Theorem~\ref{thm:mcgavg-cor}.
In Section~\ref{sec:IrredPropF} 
we prove irreducibility and property F for Ising-type MFCs. In Section~\ref{sec:physics-review} discuss these mathematical results in the context of 3D\,gravity.

\subsubsection*{Acknowledgements}
We would like to thank Ga\"etan Borot and C\'esar Galindo for useful discussions.
IRo is currently supported by the Simons Collaboration on Global Categorical Symmetry. The majority of this work was carried out as part of IRo's PhD research \cite{Romaidis-Thesis} during which IRo was supported by the Deutsche Forschungsgemeinschaft
(DFG, German Research Foundation) under Germany’s Excellence Strategy - EXC 2121 “Quantum Universe” - 390833306.
IRu is supported in part by EXC 2121 “Quantum Universe” - 390833306.

\subsubsection*{Conventions}
Throughout this text, unless otherwise specified $\catname C$ will always be a modular fusion category over an algebraically closed field $\mathbbm{k}$ of characteristic $0$. Moreover, $I$ will denote a set of representatives of simple objects in the associated modular fusion category. 

\section{Reshetikhin-Turaev TQFT and mapping class group representations}\label{sec:RT+MCG}

In this section we introduce Reshetikhin-Turaev TQFT and its underlying modular functor. We review the notions of mapping class groups and their explicit (projective) action on the Reshetikhin-Turaev state spaces. 

\subsection{Reshetikhin-Turaev topological quantum field theories}\label{sec:RT}

The Reshetikhin-Turaev (RT) TQFT is a construction of a 3-dimensional TQFT associated to an MFC $\catname{C}$. In this section, we recall the geometric input of this TQFT and the underlying modular functor, in particular the decorated bordism category $\operatorname{Bord}^{\catname{C}}$ and the mapping class groupoid $\Mod^\catname{C}$ following \cite[Chapter IV]{Tur}. 

\medskip

A \textit{connected decorated surface}, or \textit{connected d-surface}, is a (smooth, oriented, closed)\footnote{For the purpose of this text, we will
always assume that a surface is smooth, oriented and closed unless specified otherwise.
Smoothness is not a strictly stronger condition here, as for $n\leq 3$ topological and smooth $n$-manifolds are equivalent. 
}
surface $\Sigma$ together with a finite set of framed points. A framed point is a point $p\in \Sigma$ equipped with a tangent vector $v_p$, a label $X\in \catname{C}$ and a sign $\epsilon_p \in \{\pm\}$. More generally, a d-surface is a disjoint union of connected d-surfaces.
The \textit{negation} of a d-surface $\Sigma$ is the d-surface obtained by orientation reversal $-\Sigma$ and replacing the tangent vectors by $-v_p$ and signs by $-\epsilon_p$. For simplicity we often leave signs unspecified for which we assume every sign to be positive. The negation of such labels will result to the corresponding dual objects. 
A diffeomorphism of d-surfaces is a diffeomorphism between the underlying surfaces, which in addition preserves the framed points. Two such diffeomorphisms are \textit{isotopic}, if they are isotopic through diffeomorphisms of decorated surfaces. 

A \textit{decorated 3-manifold} is a smooth, oriented, compact 3-dimensional manifold with an embedded $\catname{C}$-coloured ribbon graph such that all ribbons end on coupons or the boundary $\partial M$. 
The colouring and framing of the ribbon graph give the boundary $\partial M$ the structure of a decorated surface. A \textit{decorated bordism}, or \textit{d-bordism}, $\Sigma \rightarrow \Sigma'$ between two decorated surfaces is a decorated 3-manifold $M$ together with a boundary parametrisation, \ie a d-diffeomorphism $\partial M \xrightarrow{\simeq} -\Sigma \sqcup \Sigma'$. The d-surfaces $\Sigma$ and $\Sigma'$ are referred to as the \textit{ingoing} respectively \textit{outgoing} boundaries of $M$. Two d-bordisms are \textit{equivalent} if there is a diffeomorphism compatible with ribbon graphs and boundary parametrisations. 

The d-bordism category $\operatorname{Bord}^{\catname C}$ is the symmetric monoidal category formed by d-surfaces as objects and classes of d-bordisms as morphisms. The monoidal product is given by the disjoint union and the unit is the empty set $\emptyset$.

\medskip

The construction of Reshetikhin-Turaev TQFT (RT TQFT) assigns to any d-surface $\Sigma$ a finite dimensional vector space $V^\catname{C}(\Sigma)$, called the \textit{state space}, and to a d-bordism $(M;\Sigma,\Sigma',\varphi)$ it assigns a linear map
\be\label{eq:operator-inv}
\mathcal{Z}^{\catname C}(M;\Sigma,\Sigma',\varphi): V^{\catname C}(\Sigma)\rightarrow V^{\catname C}(\Sigma')~,
\ee 
between the respective state spaces. In the tuple $(M;\Sigma,\Sigma',\varphi)$ the element $\varphi$ denotes the boundary parametrisation $\varphi: \partial M \rightarrow -\Sigma\sqcup \Sigma'$. The RT TQFT  turns out to have gluing anomalies meaning that the assignment \eqref{eq:operator-inv} preserves the gluing of bordisms only up to scalars, \ie functoriality holds only projectively. In order to fix that and produce a symmetric monoidal functor, one introduces the extended version of the bordism category.

The extended bordism category $\widehat{\operatorname{Bord}}{}^\catname{C}$ has objects of the form $(\Sigma,\lambda)$ where $\Sigma \in \operatorname{Bord}^\catname{C}$ and $\lambda \subset H_1(\Sigma;\mathbb{R})$ is a Lagrangian subspace. Its morphisms are given by pairs $(M,n)$ where $M$ is a morphism in $\operatorname{Bord}^\catname{C}$ and $n$ is an integer, also referred to as the weight. This gives an anomaly-free TQFT as a symmetric monoidal functor 
\be\label{eq:RTextended}
\mathcal{Z}^{\catname C}: \widehat{\operatorname{Bord}}{}^\catname{C}\rightarrow \operatorname{Vect}~.
\ee

We will mostly focus on one particular feature
of RT TQFTs, namely
that of mapping class group representations. 
The \textit{mapping class groupoid} $\Mod^\catname{C}$ is the symmetric monoidal groupoid consisting of decorated surfaces as objects and isotopy classes of diffeomorphisms of decorated surfaces as morphisms. The mapping class groupoid can be considered as the subgroupoid of the bordism category via the symmetric monoidal functor 
\be \label{eq:mod-in-bord}
M: \Mod^{\catname C}\hookrightarrow \operatorname{Bord}^{\catname C} \ ,
\ee 
which acts as the identity on objects and to a mapping class $f$ it assigns the bordism 
\be\label{eq:mapping-cylinder}
M(f):= (\Sigma\times [0,1];\Sigma,\Sigma, f\sqcup \id{}) \ ,
\ee
where $f$ appears in the boundary parametrisation by ``twisting" the ingoing boundary. 
The bordism $M(f)$ is called the \textit{mapping cylinder} of $f$ as it consists of the underlying cylinder over $\Sigma$ and the mapping class $f$, which appears in the boundary parametrisation. 

Restricting the RT TQFT on the mapping class groupoid via the functor in \eqref{eq:mod-in-bord} yields the \textit{modular functor}. 
However, as with the TQFT, functoriality holds only projectively due to the gluing anomalies. Similarly, one introduces the extended mapping class groupoid $\widehat{\Mod}^{\catname C}$ with the same objects as $\widehat{\operatorname{Bord}}^{\catname C}$ and pairs $(f,n)$ as morphisms, where $f$ is a mapping class and $n$ is an integer. The outcome is now a symmetric monoidal functor 
\be\label{eq:ex-modular-functor}
V^{\mathcal{C}}: \widehat{\Mod}^{\catname C}\rightarrow \operatorname{Vect}~.
\ee
We will often omit writing $\mathcal{C}$ in the superscript when it is clear from the context. 

An important ingredient of the modular functor is formed by the gluing isomorphisms which describe how the state spaces behave under cutting and gluing of surfaces. Let $\gamma$ be a simple closed curve on a d-surface $\Sigma$. By cutting along $\gamma$ we obtain a surface $\Sigma\backslash\gamma$ which has two boundary circles $\gamma^{(1)}$ and $\gamma^{(2)}$ obtained from cutting. Denote by $\Sigma^\gamma(i,j)$ the d-surface obtained from filling the holes in $\Sigma\backslash\gamma$ with disks and inserting in the middle of each disk a framed point labelled by an object $i$ for the $\gamma^{(1)}$ component and an object $j$ for the $\gamma^{(2)}$ component, see Figure~\ref{fig:mf-cutting}.

\begin{figure}[tb]
    \centering
    \begin{equation*}
    	\boxpic{1}{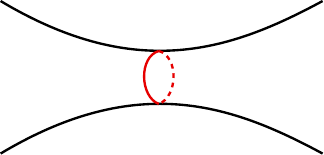}{
    		\put (10,20) {$\Sigma$}
    	    \put (47,10) {$\gamma$}
    }\quad 
\rightsquigarrow
\quad 
\boxpic{1}{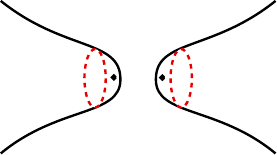}{
\put (52,25) {$j$}
\put (45,25) {$i$}
}
    \end{equation*}
    \caption{Cutting a d-surface $\Sigma$ along a simple closed curve $\gamma$ and the resulting d-surface obtained by inserting labels $i$ and $j$ in the filled holes.}
    \label{fig:mf-cutting}
\end{figure}

\begin{figure}[tb]
\centering
\boxpic{1}{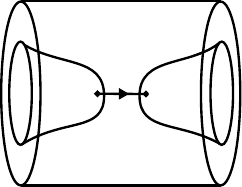}{
\put (50,43) {$i$}
}
\caption{A local picture of $N^\gamma_i$ around the cut along $\gamma$.}
\label{fig:Gluing-iso-3-mfd}
\end{figure}

Then, there are gluing isomorphisms 
\be\label{eq:mf-gluing-iso}
g_{\gamma,\Sigma} = \bigoplus_{i\in I}{g_i}: ~\bigoplus_{i\in I}{V(\Sigma^\gamma(i^\ast,i))} \rightarrow V(\Sigma) \ ,
\ee
which are natural in $\Sigma$ and compatible with the symmetric monoidal structure. These isomorphisms are provided by the RT TQFT as follows: Define the 3-manifold 
\be 
N^\gamma_i := (\Sigma^\gamma(i^\ast, i) \times [0,1])/\sim \ ,
\ee
where $\sim$ identifies on $\Sigma^\gamma(i^\ast,i)\times \{1\}$ the glued-in disks, therefore resulting in $\Sigma$. The manifold $N^\gamma_i$ has an incoming boundary at $0$ which is $\Sigma^\gamma(i^\ast,i)$ and an outgoing boundary at $1$ being $\Sigma$. The two framed points on the ingoing boundary are connected via an $i$-labelled ribbon in the interior of $N^\gamma_i$, see Figure~\ref{fig:Gluing-iso-3-mfd} and \cite[Sec.~2.6]{FRS5} for more details. Therefore, it defines a bordism 
\be 
N^\gamma_i: \Sigma^\gamma(i^\ast,i) \rightarrow \Sigma
\ee
in $\operatorname{Bord}^\mathcal{C}$. The TQFT evaluation on this bordism defines the $i$'th summand $g_i$ in \eqref{eq:mf-gluing-iso}, \ie $g_i= \mathcal{Z}^\mathcal{C}(N^\gamma_i)$.

For a d-surface $\Sigma$, the group of automorphisms $\Mod^{\catname C}(\Sigma):=\Mod^{\catname C}(\Sigma, \Sigma)$ is called the \textit{mapping class group} (MCG) of $\Sigma$. 
This is closely related to the ordinary geometric definition of mapping class groups as in \cite[Chapter 2]{FM}. For instance, let $(X,\epsilon) \in ob(\catname C)\times \{\pm\}$ be the label of every framed point on a d-surface $\Sigma$. Then, $\Mod^\catname{C}(\Sigma)$ corresponds to the framed mapping class group $\Mod(\Sigma)$ as in \cite[Chapter 2]{FM}. Here, we include a superscript of $\catname{C}$ to indicate that framed points carry labels. In general, $\Mod^{\catname C}(\Sigma)$ will be a subgroup of $\Mod(\Sigma)$ as its elements are required to also preserve labels. 

Furthermore, consider the subgroup of diffeomorphisms (up to isotopy), which fix each framed point pointwise. This subgroup is called the (framed) \textit{pure mapping class group} and we denote this by $\PMod(\Sigma)$, where we omit the superscript of $\catname C$ as such diffeomorphisms preserve by definition the label of each framed point. The mapping class group $\Mod(\Sigma_{g,n})$ of a surface $\Sigma_{g,n}$ with genus $g$ and $n$ framed points will be also denoted by $\Mod_{g,n}$ for short and the associated pure mapping class group by $\PMod_{g,n}$.
In fact, there is a short exact sequence 
\be\label{eq:ses-non-pure}
1\rightarrow \PMod_{g,n} \rightarrow \Mod_{g,n} \rightarrow S_n\rightarrow 1 \ ,
\ee
where the first arrow is the inclusion of $\PMod_{g,n}$ in $\Mod_{g,n}$ and the second arrow is determined the restriction of mapping classes onto the set of framed points. 
The \textit{unframed} version is obtained by simply omitting the data of tangent vectors in the description above. 
The unframed pure mapping class group $\PMod_{g,n}^\mathrm{un}$ is finitely generated by Dehn twists around the simple closed curves shown in Figure~\ref{fig:mcggenerators-w-insertions} \cite[Section 4.4.4]{FM}. A generating set for the framed pure mapping class group $\PMod_{g,n}$ is then obtained by adding Dehn twists around each point to the generating set. 
\begin{figure}[tb]
    \centering
    \boxpic{0.7}{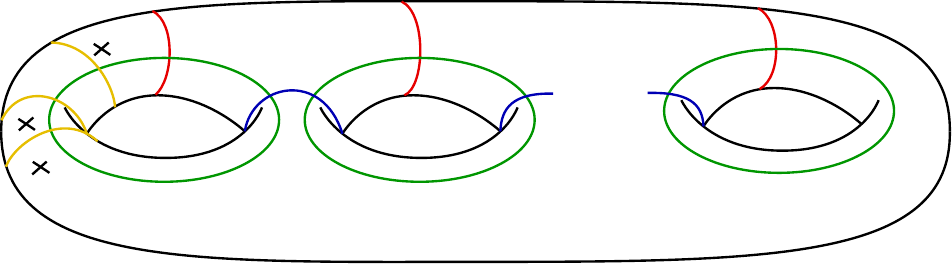}{
        \put (17,27.5) {$\alpha_1$}
        \put (43,28.5) {$\alpha_2$}
        \put (81,27.5) {$\alpha_g$}
        \put (15,5.5) {$\beta_1$}
        \put (43,5.5) {$\beta_2$}
        \put (82,6) {$\beta_g$}
        \put (30,19.5) {$\gamma_1$}
        \put (56,19.5) {$\gamma_2$}
        \put (67,19.5) {$\gamma_{g-1}$}
        \put (61,17.5) {$\dots$}
        \put (-3,8) {$\delta_1$}
        \put (-4,15) {$\delta_2$}
        \put (0,25) {$\delta_{n-1}$}
        \put (4,18.5) {$\iddots$}
    }
    \caption{Generator curves of the (unframed) pure mapping class group.}
    \label{fig:mcggenerators-w-insertions}
\end{figure}
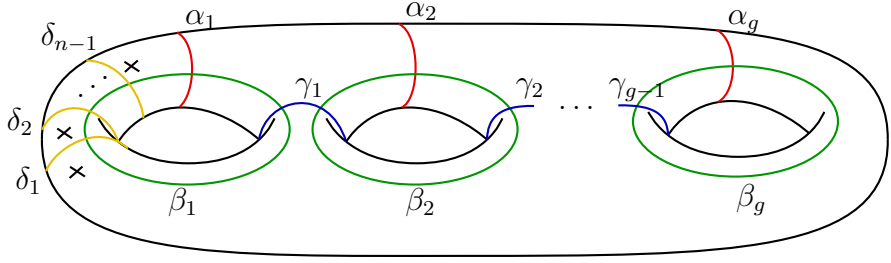
\begin{example}\label{ex:braid-group}
Genus 0 mapping class groups are directly related to (framed) braid groups. For labels $X_1,\dots, X_n$ in $\catname C$ one may define the braid group with such labels as the mapping class group of the disk 
\be\label{eq:braid-group-labels} 
\operatorname{B}^{\catname{C}}_n(X_1,\dots, X_n) := \Mod(\mathbb{D}_n^2(X_1,\dots, X_n), \partial) \ ,
\ee
where $\mathbb{D}_n^2(X_{1},\dots, X_n)$ denotes the unit disk with $n$ framed points labelled by $X_i$'s. Mapping classes are also required to fix pointwise the boundary. 
By forgetting labels and framings, one obtains the usual definition of the braid group (for more on this example see \cite[Ex.\ 3.1]{Romaidis-Thesis}). 
\end{example}

From now on, for simplicity we will only work with the non-extended mapping class groups as the projectivity of the associated representations will not pose a problem in our discussions. Notice that the extended mapping class group relates to the ordinary via the short exact sequence
\be\label{eq:ses-extendedMCG}
0\rightarrow \mathbb{Z} \rightarrow \widehat{\Mod}_{g,n}\rightarrow \Mod_{g,n}\rightarrow 1 \ ,
\ee
where the first arrow maps $n\in\mathbb{Z}$ to the weighted identity mapping class $(\id{},n)$ and the second arrow is the projection onto the first factor $(f,n)\mapsto f$.

\begin{remark}\label{rem:anomalies}
	For an MFC $\mathcal{C}$ define the scalars 
	\begin{equation}\label{eq:p_pm}
		p_\pm := \sum_{i\in I}{\theta_i^{\pm 1} d_i^2}
  \ ,
 \end{equation}
where $\theta_i$ and $d_i$ denote the twist eigenvalue respectively the quantum dimension of the object $i\in I$. 
The projective factors appearing from the gluing anomalies are half-integer powers of the so-called \textit{anomaly factor} $p_+/p_-$.
Importantly, the anomaly factor $p_+/p_-$ of an MFC is always a root of unity.

If $p_+/p_-=1$ then $\mathcal{C}$ is called \textit{anomaly-free} and its RT TQFT resp.\ its modular functor descend to symmetric monoidal functors on the non-extended bordism category resp.\ non-extended mapping class groupoid. 

The Drinfeld centre $\mathcal{Z}(\mathcal{C})\simeq \mathcal{C}\boxtimes \mathcal{C}^\mathrm{rev}$ of any MFC is automatically anomaly-free. Its RT TQFT is of Turaev-Viro type in the sense that $\mathcal{Z}^{\mathcal{Z}(\mathcal{C})}\cong \mathcal{Z}^{\mathrm{TV},\mathcal{C}}$ 
(see \cite{KirBalsam, TuVi2}) 
with the latter being the Turaev-Viro TQFT of $\mathcal{C}$ \cite{TViro,BW}. 
\end{remark}

Given a d-surface $\Sigma$, its \textit{double} is the d-surface 
\be 
\widehat\Sigma = \Sigma \sqcup -\Sigma~.
\ee
Naturally, the associated state space $V^{\catname C} (\widehat\Sigma)$ carries a projective $\Mod(\widehat\Sigma)$-action. Notice that $\Mod(\Sigma)$ embeds diagonally in $\Mod(\widehat\Sigma)$ via $f \mapsto f \sqcup f$. 
Hence, $V^\mathcal{C}(\widehat\Sigma)$ restricts to $\Mod(\Sigma)$-representation which is in fact non-projective representation (the projective factors cancel out). In Section~\ref{sec:conformal-corr} we will consider the space of invariants under this diagonal action, \ie
\begin{equation}\label{eq:hat-Sigma-inv}
    V^\mathcal{C}(\widehat\Sigma)^{\Mod(\Sigma)}~.
\end{equation}

An important property of the modular functor associated to the RT TQFT is that there exists a non-degenerate pairing
\be 
\label{eq:selfdualpairing}
\pairing_\Sigma: V^{\catname C}(\Sigma) \otimes V^{\catname C}(-\Sigma)\rightarrow \mathbbm{k} 
\ee
(see \cite[Eq.~III.(1.2.4)]{Tur} and \cite[Thm.~III.2.1.1]{Tur}). Here, natural means that the pairing is compatible with d-diffeomorphisms. In particular, it is invariant under the action of the mapping class group, \ie for any $(x,y) \in V^{\catname C}(\Sigma) \times V^{\catname C}(-\Sigma)$ and any mapping class $f\in \Mod(\Sigma)$
\be
\label{eq:pairing-inv}
\pairing_\Sigma(f.x,f.y) ~=~ \pairing_\Sigma(x,y)~.
\ee
In fact, $\pairing_\Sigma$ is the result of applying the RT TQFT on the cylinder $\Sigma\times [0,1]$ seen as a bordism $\Sigma\sqcup -\Sigma \rightarrow \emptyset$. For a mapping class $f$ on $\Sigma$ the diffeomorphism $f\times \id{[0,1]}$ does not change the equivalence class of $\Sigma\times [0,1]$, which implies invariance of the pairing under mapping class group action.

It will be convenient to describe the vector spaces $V^\catname{C}(\Sigma)$ explicitly in terms of morphism spaces. 
Let $\Sigma$ be a d-surface and fix an ordering of its labels $(X_1,\epsilon_1), \dots, (X_n,\epsilon_n)$. The state space $V(\Sigma_{g,n})$ is isomorphic to the morphism space 
\be\label{eq:V_g,n-def}
V_{g,n} ~=~ \catname{C}(\unit, X_1^{\epsilon_1}\otimes\dots\otimes X_n^{\epsilon_n}\otimes L^{\otimes g}) \ ,
\ee
where  $X^+\equiv X$, $X^- \equiv X^\ast$ and $L=\bigoplus_{i\in I}{i\otimes i^\ast}$.
We decompose $V_{g,n}$ into a direct sum as
\begin{equation}\label{eq:V_i-decomp}
    V_{g,n} ~=~ \bigoplus_{i \in I^g} V_i \ ,
\end{equation}
where $V_{i} = \mathcal{C}(\unit, X_1^{\epsilon_1}\otimes\dots\otimes X_n^{\epsilon_n}\otimes \bigotimes_{k=1}^g {i_{k} \otimes i_k^\ast})$ for $i=(i_{1}, \dots, i_g) \in I^g$. 
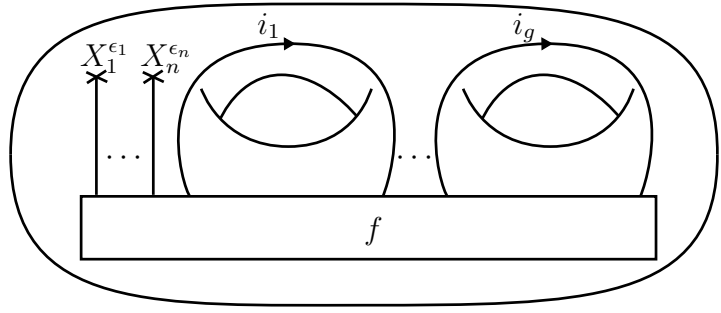
\begin{figure}[tb]
    \centering
    \begin{overpic}[scale = 0.9]{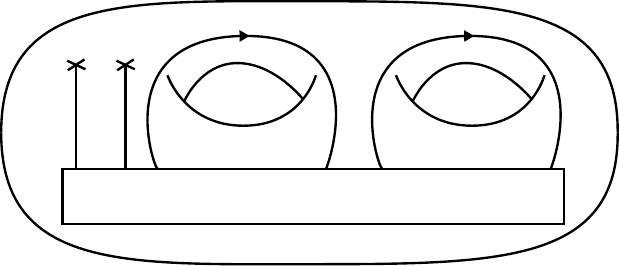}
        \put (50,9.5) {$f$}
        \put (10,34) {$X_1^{\epsilon_1}$}
        \put (18.5,34) {$X_n^{\epsilon_n}$}
        \put (13.5, 20) {$\cdots$}
        \put (54.5,20) {$\cdots$}
        \put (35,38.5) {$i_1$}
        \put (71,38.5) {$i_g$}
    \end{overpic}
    \caption{A handlebody with embedded ribbon graph and boundary $\Sigma_{g,n}$. The coupon is labelled by a morphism
    $f \in V_i \subset \catname{C}(\unit, X_1^{\epsilon_1}\otimes\dots\otimes X_n^{\epsilon_n}\otimes L^{\otimes g})$.
    }
    \label{fig:handlebodycoupon}
\end{figure}
For a vector $f\in V_i$, consider the handlebody with an embedded ribbon graph in Figure \ref{fig:handlebodycoupon}, where the $k$'th ribbon strand is directed upwards respectively downwards if $\epsilon_k = +$ respectively $\epsilon_k = -$ and the coupon is labelled by the morphism $f$. In this figure, the framing is the one where all tangent vectors point to the right. This represents a d-bordism $\emptyset\rightarrow \Sigma_{g,n}$. 
Evaluating the RT TQFT on this bordism provides an isomorphism $V_{g,n} \cong V(\Sigma_{g,n})$. 

Let us describe the pairing in terms of morphism spaces as in \eqref{eq:V_g,n-def}. Fixing the isomorphism between $V(\Sigma_{g,n})$ and \eqref{eq:V_g,n-def} fixes an isomorphism for $V(-\Sigma_{g,n})$ and 
\begin{equation}
    \overline{V}_{g,n} ~=~ \mathcal{C}(X_1^{\epsilon_1}\otimes\dots\otimes X_n^{\epsilon_n}\otimes L^{\otimes g},\unit) ~=~ V_{g,n}^\ast \ .
\end{equation}
The pairing between a morphism $f \in V_{g,n}$ and a morphism $g\in \overline{V}_{g,n}$ is then given by 
\begin{equation}
    d(f,g) = g\circ f \ ,
\end{equation}
also known as the trace pairing (see \cite[Sec.\ IV.1.7]{Tur} and note that the special ribbon graph representing the cylinder bordism is given in \cite[Fig.\ IV.2.4]{Tur}). 
In particular, for $f_{i} \in V_{i}$ and $g_{j} \in \overline{V}_j$ we have 
\begin{equation}
    d(f_i,g_j) \neq 0 \quad \Leftrightarrow \quad i = j,~g_i \circ f_i \neq 0 \ .
\end{equation}

We describe this for a specific family of vectors in $V_i$. Let $T$ be a morphism in $V_i$ which is obtained by a ribbon tangle with no link components, \ie it is represented by a ribbon graph with no coupons and no loops and whose colours are determined by the outgoing labels, see Figure~\ref{fig:tangle-handlebody} for an example. In particular, if $T$ is non-zero then $n$ is even since for each $X_i^{\epsilon_i}$ the label $X_{i}^{-\epsilon_i}$ is also contained in the label set. 

Reflecting along the horizontal plane and reversing orientations for each strand gives rise to a morphism in $\overline{V}_i$ denoted by $-T$. The tensor product $\widehat{T}:=T\otimes -T $ represents an element in $V(\widehat\Sigma)$. 

\begin{lemma}\label{lem:tangle-pairing}
    Let $T$ be a morphism in $V_{i}$ represented by a $\mathcal{C}$-coloured ribbon tangle with no link components. Then we have 
    \begin{equation}
        \pairing_\Sigma (\widehat{T}) \neq 0~. 
    \end{equation}
\end{lemma}
\begin{proof}
    Without loss of generality we may write $T$ as a composition $T = T_t \circ T_b \circ T_c$,
    where $T_t$ is represented by a ribbon tangle whose underlying unframed tangle is the identity, $T_b$ is a braid with no framings and $T_c$ consists of a tensor product of coevaluation maps. 
    The pairing is given by 
    \[ -T \circ T = -T_c \circ -T_b \circ -T_t \circ T_t \circ T_b \circ T_c = -T_c \circ T_c \neq 0 \ ,\]
    where the second equality follows from definition and $-T_c\circ T_c$ is non-zero as it is a product of
    quantum dimensions, all of which are non-zero.
\end{proof}

\begin{figure}[tb]
    \centering
    \boxpic{0.8}{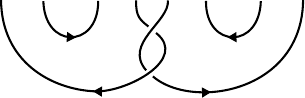}{
    \put (-1,34) {$i$}
    \put (12,34) {$i^\ast$}
    \put (30,34) {$i$}
    \put (43,34) {$i^\ast$}
    \put (53,34) {$j^\ast$}
    \put (65,34) {$j$}
    \put (83,34) {$j^\ast$}
    \put (97,34) {$j$}
    }
    $\rightsquigarrow$
    \boxpic{0.8}{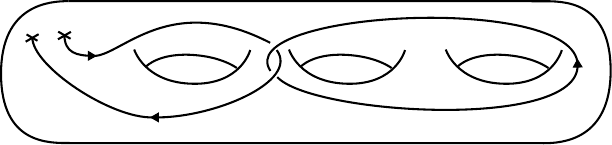}{ 
    \put (95,10) {$j$}
    \put (10,5) {$i$}
    }
    \caption{A tangle morphism in $V_{3,2}$ with labels $i$ and $i^\ast$ which represents a vector in $V(\Sigma_{3,1})$ via the associated handlebody on the right.}
    \label{fig:tangle-handlebody}
\end{figure}
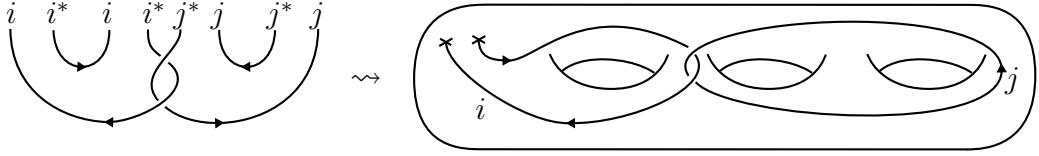

\begin{example}\label{ex:handlebody-pairing}
	Let $\Sigma$ be a d-surface without framed points and let $H_\Sigma$ be the handlebody bounding $\Sigma$. Evaluating the TQFT on $H_\Sigma$ seen as a bordism $\emptyset \rightarrow \Sigma$ gives a vector 
	\begin{equation}
		x:=\mathcal{Z}^\mathcal{C}(H_\Sigma)\in V^\mathcal{C}(\Sigma) \ .
	\end{equation}
 Similarly, evaluating the handlebody with reversed orientation $-H_\Sigma: \emptyset \rightarrow -\Sigma$ gives a vector
 \begin{equation}
 	\hat{x}:= \mathcal{Z}^\mathcal{C}(-H_\Sigma)\in V^\mathcal{C}(-\Sigma)~.
  \end{equation}  
By Lemma~\ref{lem:tangle-pairing} for the empty tangle $T = \emptyset$ we have 
\begin{equation}\label{eq:x-hatx=1}
	\pairing_\Sigma(\hat{x}) \neq 0~.
\end{equation}
\end{example}

\subsection{Mapping class group action on RT state spaces}\label{subsec:MCGaction}

We will only need explicit expressions for the action of the pure mapping class group.
Fix a surface $\Sigma_{g,n}$ of genus $g$ and $n$ framed points labelled by simple objects $l_1,\dots, l_n \in I$. 
For a simple closed curve $\gamma$ on $\Sigma_{g,n}$ one can define the Dehn twist $T_\gamma$ along $\gamma$ as a mapping class in $\PMod_{g,n}$. It is obtained by cutting out an annular neighbourhood of $\gamma$, performing a full twist and gluing it back to the surface. A set of generators of the pure mapping class group is given as the set of Dehn twists along the curves in Figure~\ref{fig:mcggenerators-w-insertions} and also (due to the framings) Dehn twists $T_{\lambda_k}$ around each framed point, see \cite[Corollary~4.15]{FM}. 
We further define the so-called $S$-transformations by $S_k := T_{\alpha_k}\circ T_{\beta_k}\circ T_{\alpha_k}$ and we replace in our generating set the Dehn twists $T_{\beta_k}$ by the corresponding $S$-transformations $S_k$.

Following \cite[Ch.\,IV]{Tur}
we can now describe explicitly how these generators act on vectors of the state space $V^{\catname C}(\Sigma_{g,n})$ up to projectivity. Recall from \eqref{eq:V_g,n-def} that this vector space is identified with the morphism space 
\begin{equation}\label{eq:hom-space-state-space}
\catname C(\unit, l_1\otimes \cdots \otimes l_n \otimes L^{\otimes g} ) \ ,
\end{equation}
which in turn
decomposes into a direct sum of $V_i$'s as in \eqref{eq:V_i-decomp}.
For 
$i = (i_1,\dots, i_g) \in I^g$ and
a vector $f_i \in V_i$, the mapping class group generators act\footnote{The action is projective with projective factors given by half-integer powers of $p_+/p_-$ (Remark \ref{rem:anomalies}).}
as follows: 
\vspace{1ex}
\begin{align}\label{eq:MCG-gen-action}
T_{\lambda_k}(f_i) = \theta_{l_k} &\boxpic{0.7}{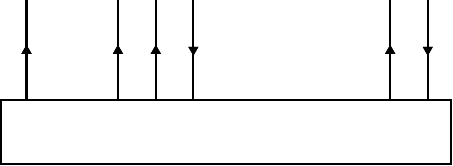}{
        \put (3,39) {$l_1$}
        \put (23,39) {$l_n$}
        \put (32,39) {$i_1$}
        \put (42,39) {$i_1^\ast$}
        \put (85,39) {$i_g$}
        \put (94,39) {$i_g^\ast$}
        \put (45,5) {$f_i$}
        \put (10,25) {$\cdots$}
        \put (60,25) {$\cdots$}
        }
&
T_{\delta_k}(f_i) =
    \boxpic{0.7}{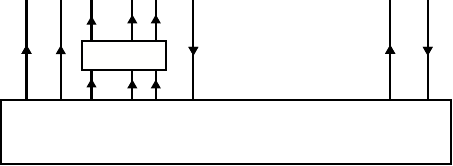}{
        \put (2,39) {$l_1$}
        \put (10,39) {$l_k$}
        \put (17,39) {\tiny$\cdots$}
        \put (26.5,39) {$l_n$}
        \put (32,39) {$i_1$}
        \put (42,39) {$i_1^\ast$}
        \put (85,39) {$i_g$}
        \put (94,39) {$i_g^\ast$}
        \put (45,5) {$f_i$}
        \put (25,22) {\small$\theta$}
        \put (60,25) {$\cdots$}
        }
\nonumber\\[2ex]
T_{\alpha_k}(f_i) = \theta_{i_k} &\boxpic{0.7}{figures/f_i.pdf}{
        \put (2,39) {$l_1$}
        \put (23,39) {$l_n$}
        \put (32,39) {$i_1$}
        \put (42,39) {$i_1^\ast$}
        \put (85,39) {$i_g$}
        \put (94,39) {$i_g^\ast$}
        \put (45,5) {$f_i$}
        \put (10,25) {$\cdots$}
        \put (60,25) {$\cdots$}
        }
&
T_{\gamma_k}(f_i) =
    \boxpic{0.7}{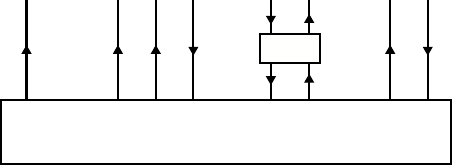}{
        \put (2,39) {$l_1$}
        \put (23,39) {$l_n$}
        \put (32,39) {$i_1$}
        \put (42,39) {$i_1^\ast$}
        \put (85,39) {$i_g$}
        \put (94,39) {$i_g^\ast$}
        \put (45,5) {$f_i$}
        \put (56,39) {$i_k^\ast$}
        \put (65,39) {$i_{k+1}$}
        \put (62,23.5) {\small$\theta$}
        \put (10,25) {$\cdots$}
        \put (47,25) {$\cdots$}
        \put (73,25) {$\cdots$}
        }
\nonumber\\[2ex]
S_k(f_i)=\bigoplus_{j\in I} \frac{d_j}{D}
    &\boxpic{0.7}{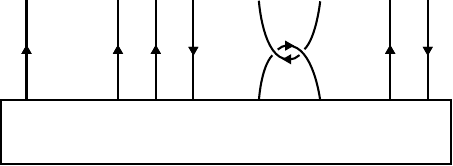}{
        \put (2,39) {$l_1$}
        \put (23,39) {$l_n$}
        \put (32,39) {$i_1$}
        \put (42,39) {$i_1^\ast$}
        \put (85,39) {$i_g$}
        \put (94,39) {$i_g^\ast$}
        \put (45,5) {$f_i$}
        \put (54,39) {$j$}
        \put (68,39) {$j^\ast$}
        \put (10,25) {$\cdots$}
        \put (47,25) {$\cdots$}
        \put (73,25) {$\cdots$}
    }
\end{align}
Here, we used $\theta$ to denote the twist morphism in $\catname C$ and $D$ is a square root of its quantum dimension, \ie $D = \sqrt{\sum_{i\in I}{d_i^2}}$. 
Note that all generators except for $S_k$ map the direct summand $V_i$ back to itself.

\medskip

\begin{figure}[tb]
    \centering
    \boxpic{0.7}{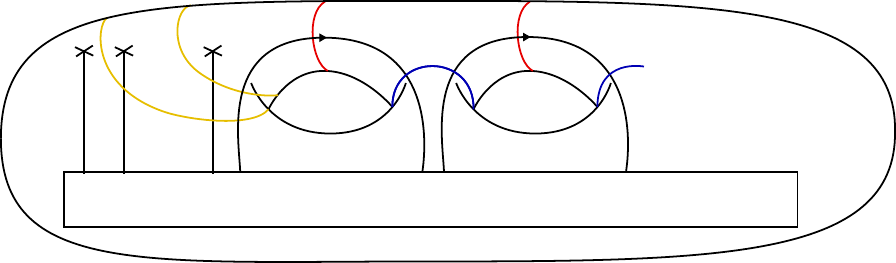}{
        \put (34,30.5) {$\alpha_1$}
        \put (58,30.5) {$\alpha_2$}
        \put (45,23.5) {$\gamma_1$}
        \put (68,23.5) {$\gamma_2$}
        \put (77,18) {$\dots$}
        \put (10,29) {$\delta_1$}
        \put (20,30) {$\delta_{n-1}$}
        \put (6.5,15) {$l_1$}
        \put (10.5,15) {$l_2$}
        \put (20,15) {$l_{n}$}
        \put (27.5,25) {$i_1$}
        \put (50,25) {$i_2$}
        \put (45,6) {$f_i$}
    }
    \caption{Arrangement of coupon and ribbons in the handlebody used to derive the mapping class group action on the Hom-space \eqref{eq:hom-space-state-space}.}
    \label{fig:mcggenerators-and-ribbongraph}
\end{figure}
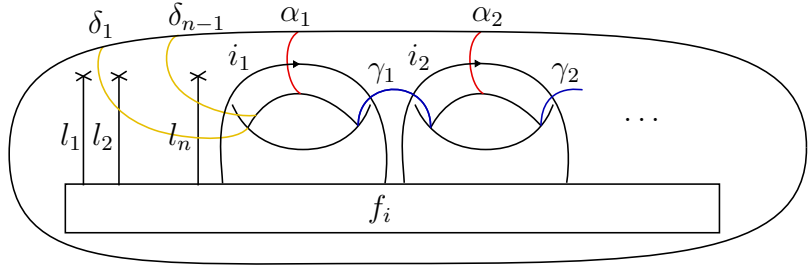
To arrive at these expressions, consider a handlebody with a coupon labelled by $f_i$ as in Figure~\ref{fig:mcggenerators-and-ribbongraph}.
The pure mapping class group acts on the state space assigned to the boundary surface via mapping cylinders as in \eqref{eq:mapping-cylinder}.
Thus the action of a mapping class $h$ on $f_i$ is obtained by gluing the mapping cylinder $M(h)$ over $\Sigma_{g,n}$ to the handlebody.

For a simple closed curve $\gamma$ on $\Sigma_{g,n}$, which is contractible with respect to the handlebody, the Dehn twist $T_\gamma$ acts by twisting the ribbons passing through the disc in the handlebody bounded by the curve. This is the origin of the twists $\theta$ in the first four equations in \eqref{eq:MCG-gen-action}.

Details on the derivation of the last equation in \eqref{eq:MCG-gen-action} can be found \eg in \cite[Def.\ 3.1.15]{BK}.

\section{Correlators as mapping class group averages}\label{sec:conformal-corr}

The goal of this section is to establish a correspondence between RCFT correlators and MCG averages. This is surprisingly motivated by physics considerations in 3D\,quantum gravity, which will be explained in the later Section~\ref{sec:physics-review}. In Section
\ref{subsec:RCFT} we briefly recall the basics of rational CFT and its RT TQFT description. In Section \ref{subsec:MCGaverage} we define mapping class group averages under a finiteness condition and in Section \ref{subsec:correspondence} we relate for certain theories RCFT correlators with MCG averages.  	
	
\subsection{RCFT correlators}\label{subsec:RCFT}

A 2-dimensional rational CFT is characterised by its rational VOA algebra $\mathcal{A}$ and its correlation functions that live on Riemann surfaces with field insertions. 
Its representation category $\catname C = \operatorname{Rep}(\mathcal{A})$  
is an MFC \cite{Huang}. The TQFT approach to RCFT in \cite{FRS1} utilises the RT TQFT associated to $\catname C$ to provide all consistent systems of correlators.

The state space $V^{\catname C} (\Sigma)$ for a surface $\Sigma$ models the space of chiral conformal blocks.
Similarly, correlators are viewed as elements in the state space of the double $\widehat\Sigma$.
We summarise the topological description of correlators: 
\begin{itemize}
    \item The correlator $\operatorname{Cor}(\Sigma)$ on a surface $\Sigma$ is an element of the state space $V^\catname{C}(\widehat\Sigma)$ of the double surface $\Sigma$, \ie 
    \be\label{eq:hol-fact} 
    \operatorname{Cor}(\Sigma) \in V^\catname{C}(\widehat\Sigma)~.
    \ee 
    This is known as \textit{holomorphic factorisation} in that correlators solve the chiral Ward identities. 
    \item The correlator $\operatorname{Cor}(\Sigma)$ is invariant under the diagonal action of the mapping class group $\Mod(\Sigma)$, \ie
    \be\label{eq:Cor-mod-inv}
    \operatorname{Cor}(\Sigma) \in V^\mathcal{C}(\widehat\Sigma)^{\Mod(\Sigma)}~.
    \ee 
    This invariance, also called modular invariance, implies that correlation functions are single-valued. 
    \item Correlators solve sewing constraints, \ie they behave nicely under cutting and gluing of a surface along some boundary circle. 
    This property is captured in terms of the modular functor by the gluing isomorphism $g$ from \eqref{eq:mf-gluing-iso}. Let $\Sigma$ be a d-surface and $\gamma$ be a simple closed curve on $\Sigma$. Then
    \be\label{eq:cor-factorisation}
    \operatorname{Cor}(\Sigma) = g_{\gamma, \widehat\Sigma}\left(\sum_{i,j,\lambda}\operatorname{Cor}(\Sigma_\gamma(i,j))\right)
    \ ,
    \ee
    where $\Sigma_{\gamma}(i,j)$ is the cut surface with label $(i,j)\in \catname C \boxtimes \catname C^\mathrm{rev}$ resp.\ its dual and $\lambda$ runs over possible multiplicities. The sum is over intermediate states which means over all simple labels $i,j\in I$. Finally, the $g_{\gamma,\widehat\Sigma}$ is the gluing isomorphism \eqref{eq:mf-gluing-iso} when gluing along $\gamma$ in both copies in $\widehat\Sigma$.
    This is also called the \textit{factorisation property} of correlators and for the sphere with four framed points this is the realisation of an operator product expansion (OPE).
\end{itemize}

The explicit construction of RCFT correlators for a given MFC $\catname C$ relies further on a symmetric special Frobenius algebra $A$ in $\catname C$. The corresponding CFT correlators depending on $A$ are denoted by 
\be 
\operatorname{Cor}^{\catname C}_A(\Sigma) \in V^{\catname C}(\widehat\Sigma) \ .
\ee
Recall that symmetric special Frobenius algebras give rise to surface defects in $\catname C$ \cite{KS,FSV,CRS2}.
We can exploit this to construct the correlators given in \cite{FRS5} in terms of the RT TQFT with defects
\begin{equation}
    \mathcal{Z}^{\mathcal{C},\mathrm{def}}: \operatorname{Bord}_3^{\mathcal{C},\mathrm{def}} \rightarrow \Vect 
\end{equation}
introduced in \cite{CRS2}. The category $\operatorname{Bord}_3^{\mathcal{C},\mathrm{def}}$ above denotes the category of bordisms which also include defects (of any codimension) with labels depending on the codimension. 

Let $\Sigma$ be a d-surface with framed points labelled by pairs $(i,j) \in \catname C\times \catname C$ which correspond to bulk field insertions. Consider the defect cylinder\footnote{Its underlying 3-manifold is the cylinder $\Sigma \times [0,1]$.} 
\begin{equation}
\widehat{M}_\Sigma : \emptyset \rightarrow \widehat\Sigma \ ,
\end{equation}
with a surface defect $\Sigma \times \{1/2\}$ labelled by the symmetric special Frobenius algebra $A$. The candidate for the RCFT correlator on $\Sigma$ is then defined as
\be\label{eq:RCFT-cor-via-def}
\operatorname{Cor}_A^{\catname C}(\Sigma) = \mathcal{Z}^{\mathcal{C},\mathrm{def}}
(\widehat{M}_\Sigma)~
\ee
by applying the defect RT TQFT on the right (Figure~\ref{fig:cor-bnd-def} on the right hand side). 

There is a dual description in terms of boundary conditions for the double theory $\doublecat$. In terms of boundary conditions consider the cylinder over $\Sigma$ as a bordism 
\begin{equation}
M_\Sigma: \emptyset \rightarrow \Sigma    
\end{equation}
by declaring $\Sigma\times \{0\}$ as a boundary condition labelled by $A$ (Figure~\ref{fig:cor-bnd-def} on the left hand side). Then,
we may also write
\be\label{eq:RCFT-cor-via-bnd}
\operatorname{Cor}_A^{\catname C}(\Sigma) = \mathcal{Z}^\mathcal{C\boxtimes C^\mathrm{rev},\mathrm{bnd}}(M_\Sigma) \in V^{\catname C \boxtimes \catname{C}^\mathrm{rev}}(\Sigma) \ ,
\ee
where the TQFT of $\catname C\boxtimes \catname C^\mathrm{rev}$ (which includes boundary conditions) is applied to the cylinder $M_\Sigma$. The two approaches \eqref{eq:RCFT-cor-via-def} and \eqref{eq:RCFT-cor-via-bnd} are illustrated by Figure~\ref{fig:cor-bnd-def}. 
For a detailed overview on defects and boundary conditions refer to \cite[Sec.\ 4]{Romaidis-Thesis}. Moreover, note that the approach of boundary conditions and the associated figure in Figure~\ref{fig:cor-bnd-def} resembles the sandwich picture of 
\cite{FreedMooreTeleman}.
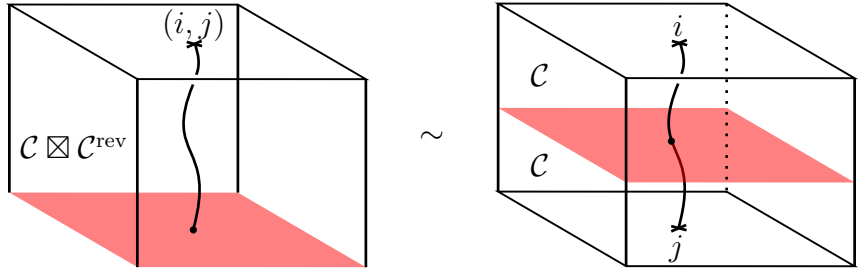
\begin{figure}[bt]
    \centering
    \begin{equation*}
    	\boxpic{0.9}{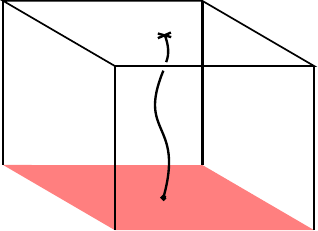}{
    		\put (4,30) {$\doublecat$}
    		\put (43,64.5) {$(i,j)$}
    	}
    	\quad \sim \quad 
    	\boxpic{0.9}{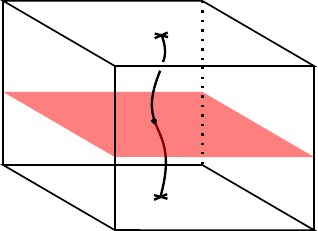}{
    		\put (10,50) {$\mathcal{C}$}
    		\put (10,25) {$\mathcal{C}$}
    		\put (49,63.5) {$i$}
    		\put (48,3) {$j$}
    	}
    \end{equation*}
    \caption{The construction of correlators via boundary in $\catname C\boxtimes \catname C^\mathrm{rev}$ or defects in $\catname C$. The picture on the left represents the cylinder over $\Sigma$ seen as a bordism from $\emptyset$ to $\Sigma$ by declaring the bottom boundary to be a boundary condition. The picture on the right represents again the cylinder, now seen as a bordism from $\emptyset \rightarrow \widehat\Sigma$ together with a defect in the middle. Correlators are obtained by evaluating with the associated TQFTs.}
    \label{fig:cor-bnd-def}
\end{figure}

The following theorem is a result from \cite{FRS5}:
\begin{theorem}
Let $\catname C$ be an MFC and $A$ a symmetric special Frobenius algebra in $\catname C$. The correlators $\{\operatorname{Cor}_A^\mathcal{C}\}_\Sigma$ defined in \eqref{eq:RCFT-cor-via-def} form a consistent system of correlators, \ie they satisfy \eqref{eq:hol-fact}, \eqref{eq:Cor-mod-inv} and \eqref{eq:cor-factorisation}.
\end{theorem}

The result is even stronger in that
it is shown
in \cite{Fjelstad:2006aw} that Morita equivalent symmetric special Frobenius algebras in $\catname C$ lead to the same correlators and thus equivalent CFTs. This forms in fact a one-to-one correspondence between 2D\,RCFTs and pairs $(\catname C, [A])$ of an MFC $\catname C$ and a Morita class of a symmetric special Frobenius algebra $A$ in $\catname C$. 

\subsection{Mapping class group averages}\label{subsec:MCGaverage}

We now define the mapping class group averages, which should be thought as candidates for quantum gravity correlators (discussed in Section~\ref{sec:physics-review}). 

For a given MFC $\catname C$ we ask for a finiteness condition on its mapping class group representations. This will allow for a well-defined notion of mapping class group average. 
\begin{definition}\label{def:propF}
We say that an MFC $\catname{C}$ \textit{has property F with respect to a d-surface} $\Sigma$ if the associated mapping class group image $V^\catname{C}(\Mod^\mathcal{C}(\Sigma))$ is a finite group. 
If $\mathcal{C}$ has property F with respect to all d-surfaces, we say that $\catname{C}$ \textit{has property F}.
\end{definition}
\begin{remark}\label{rem:propF}
    \begin{enumerate}
        \item This terminology was originally used when studying a finiteness property for braided fusion categories \cite{NR} and their braid group representations, which we adopt and use for MFCs and mapping class groups. We may also specify to only a family of surfaces, for example with respect to surfaces without framed points. For more details see Section~\ref{subsubsec:propFgenus0}.
        \item It does not matter if we define property F with respect to the whole mapping class group $\Mod^\mathcal{C}$ or just the pure mapping class group $\PMod$,
        i.e.\ the 
        following conditions are equivalent: 
	\begin{enumerate}
		\item $V^\catname{C}(\Mod^\mathcal{C}(\Sigma))$ is finite.
		\item $V^\mathcal{C}(\PMod(\Sigma))$ is finite.
	\end{enumerate}
		a) $\Rightarrow$ b): The pure mapping class group is a subgroup and therefore $ V^\catname{C}(\PMod(\Sigma)) \subset V^\catname{C}(\Mod^\mathcal{C}(\Sigma))$. 

		b) $\Rightarrow$ a): The mapping class group $\Mod^\mathcal{C}(\Sigma)$ is an extension of $\PMod(\Sigma)$ by a subgroup $H$ of the symmetric group $S_n$, cf.\ \eqref{eq:ses-non-pure}, where $n$ is the number of framed points and $H$ is determined by allowed permutations of framed points. In particular $H$ is a finite group and $V^\mathcal{C}(\Mod^\catname{C}(\Sigma))$ is finite as an extension of a finite group by a finite group.
		
		Therefore, we will not distinguish between the two conditions and work with the pure mapping class group instead.  
	\item Let $T_1,\dots, T_n$ denote the Dehn twists around each framed point in $\PMod_{g,n}$. Such Dehn twists act on the state space by scalar multiplication of $\theta_l$ where $l$ is the label of the corresponding point, see \eqref{eq:MCG-gen-action}. Since MFCs have twist eigenvalues of finite order, the image of the subgroup $\langle T_1,\dots,T_n\rangle$ will always be finite. Hence, it is sufficient to check property F with respect to the remaining generators which are Dehn twists around the curves in Figure~\ref{fig:mcggenerators-w-insertions}, corresponding to the unframed mapping class group $\PMod_{g,n}^\mathrm{un}$.
    \end{enumerate}
\end{remark}

If an MFC has property F, and in particular mapping class group orbits are finite we give the following definition of a mapping class group average.
\begin{definition}\label{def:mcgaverage}
Let $\Sigma$ be a d-surface such that the representation image $G:=V^{\catname C}(\Mod(\Sigma))$ is finite in $\End(V^{\catname C}(\Sigma))$. Then, define the linear map 
\begin{equation}
\langle - \rangle_\Sigma: V^{\catname C}(\widehat \Sigma) \rightarrow V^{\catname C}(\widehat\Sigma),~ x \mapsto \langle x \rangle_{\Sigma} := \frac{1}{|G|}\sum_{g\in G}{g.x} \ ,
\end{equation}
where $\Mod(\Sigma)$ acts diagonally on $V^{\catname C}(\widehat\Sigma)$.
\end{definition}
The image of $\langle - \rangle_\Sigma$ is the subspace of $\Mod(\Sigma)$-invariants as in \eqref{eq:hat-Sigma-inv}, \ie 
\be 
\operatorname{Im}(\langle - \rangle_\Sigma) = V^\mathcal{C}(\widehat\Sigma)^{\Mod(\Sigma)}~.
\ee 
\begin{remark}\label{rem:mcg-avg-def}
\begin{enumerate}
	\item One can express the mapping class group average as an orbit sum 
	\be\label{eq:orbit-sum}
	 \langle x \rangle_\Sigma := \frac{1}{|G|}\sum_{g\in G}{g.x} = \frac{1}{|G.x|}\sum_{y\in G.x}{y} \ ,
	\ee
	where $G.x$ is the mapping class group orbit of $x$.
	In fact to define $\langle x \rangle_\Sigma$ it is sufficient to require finiteness of its orbit and not of the representation image. To go beyond the finiteness requirement itself requires a suitable invariant measure to integrate with. This fails for the mapping class group itself as it is non-amenable as pointed out in \cite{Jian:2019ubz}. However, it could be that the mapping class group image is amenable, even if it fails to be finite, but we are not aware of results in this direction.
	\item 
	For the torus $\Sigma = T^2$ the finiteness condition is always satisfied. This follows from a result of \cite{Ng} showing that the kernel of the associated mapping class group representation is a congruence group and thus a finite index group.
 
	\item The semisimplicity of the MFC $\catname{C}$ proves to play an important role in the finiteness condition. For non-semisimple modular tensor categories mapping class group representations generalising the semisimple case were constructed in \cite{Lyub} and extended to 3D\,TQFTs in \cite{DeRenziGGPR,DeRenziGGPR2}.
    In non-semisimple modular tensor categories, the ribbon-twist can have infinite order, and it can happen that Dehn twists along any essential closed curve have infinite order, see \cite[Prop.~5.1]{DeRenziGGPR2} for an example. This makes it harder to satisfy Property F, and in fact we are not aware of a non-semisimple modular tensor category with property F.
\end{enumerate} 
\end{remark}

\subsection{Main correspondence}\label{subsec:correspondence}

The state space of the double $\widehat\Sigma$ can be decomposed using tensoriality of the modular functor, \ie 
\be\label{eq:mf-tensoriality}
V^{\catname C}(\widehat\Sigma) ~\cong~ V^{\catname C}(\Sigma)\otimes V^{\catname C}(-\Sigma) \ .
\ee
Naturality of these isomorphisms implies in particular that they are $\Mod(\Sigma)$-intertwiners. 

By abuse of notation let $\pairing_\Sigma: V^{\catname C}(\widehat \Sigma)\rightarrow \mathbbm{k}$ denote the pullback of the pairing \eqref{eq:selfdualpairing}. If in addition to the finiteness property $\catname C$ obeys an irreducibility property, we get a relation between conformal correlators and mapping class group averages as stated in the following theorem. 
\begin{theorem}\label{thm:mcgavg-cor}
Let $\Sigma$ be a d-surface such that the projective representation $V^{\catname C}(\Sigma)$ is irreducible, $V^{\catname C}(\Mod(\Sigma)) \subset \End(V^{\catname C}(\Sigma))$ is finite and let $x\in V^{\catname C}(\widehat\Sigma)$ be an element such that $\pairing_{\Sigma}(x)\neq 0$. In addition, suppose that $A$ is a symmetric special Frobenius algebra in $\catname C$. Then, there exists $\lambda_\Sigma\in \mathbbm{k}$ such that 
\begin{equation}\label{eq:mcgavg-cor}
	 \operatorname{Cor}^{\catname C}_A(\Sigma) = \lambda_\Sigma\, \langle x \rangle_\Sigma~.
\end{equation}
\end{theorem}

The proof of the above theorem makes use of the following lemma: 

\begin{lemma}\label{lem:mcg-invariant-subsp}
	Let $\Sigma$ be a d-surface such that the projective representation $V^\mathcal{C}(\Sigma)$ is irreducible. The space of mapping class group invariants on the double $\widehat\Sigma$ is one-dimensional, \ie 
	\be
	V^\mathcal{C}(\widehat\Sigma)^{\Mod(\Sigma)} ~\cong~ \mathbbm{k}~.
	\ee 
\end{lemma}
\begin{proof}
	The non-degeneracy of the pairing $\pairing_\Sigma: V^\catname{C}(\Sigma)\otimes V^\catname{C}(-\Sigma)\rightarrow \mathbbm{k}$ gives an isomorphism $\psi:V^\catname{C}(-\Sigma)\xrightarrow{\sim} V^\mathcal{C}(\Sigma)^\ast,~y\mapsto \pairing_\Sigma(-,y)$. This is a $\Mod(\Sigma)$-intertwiner (with respect to the dual action on $V^\mathcal{C}(\Sigma)^\ast$). Indeed, for any $y\in V^\mathcal{C}(-\Sigma)$ and $f\in \Mod(\Sigma)$: 
	\[
	\psi(f.y) := \pairing_\Sigma(-,f.y) \overset{\eqref{eq:pairing-inv}}{=} \pairing_\Sigma(f^{-1}.-,y) = f.\psi(y)  \ .
	\]
	Furthermore, there is an obvious $\Mod(\Sigma)$-isomorphism $V^\mathcal{C}(\Sigma)\otimes V^\mathcal{C}(\Sigma)^\ast \cong \End(V^\mathcal{C}(\Sigma))$, which intertwines the diagonal action on the left with the conjugation action on the right. Combining this with \eqref{eq:mf-tensoriality} we obtain a $\Mod(\Sigma)$-isomorphism 
	\[V^\catname{C}(\widehat\Sigma) ~\cong~ \End(V^\mathcal{C}(\Sigma))~.\]
	In particular, they have isomorphic spaces of $\Mod(\Sigma)$-invariants. The space of $\Mod(\Sigma)$-invariants in $\End(V^\mathcal{C}(\Sigma))$ coincides with the subspace $\End_{\Mod(\Sigma)}(V^\mathcal{C}(\Sigma))$ of $\Mod(\Sigma)$-intertwiners. By Schur's Lemma, irreducibility of $V^\mathcal{C}(\Sigma)$ implies $\End_{\Mod(\Sigma)}(V^\mathcal{C}(\Sigma))\cong \mathbbm{k}$ which concludes the proof. 
\end{proof}

\begin{proof}[Proof of Theorem~\ref{thm:mcgavg-cor}]
Irreducibility of $V^{\catname C}(\Sigma)$ implies by Lemma~\ref{lem:mcg-invariant-subsp} that the space of modular invariants $V^{\catname C}(\widehat\Sigma)^{\Mod(\Sigma)}$ is 1-dimensional.
The correlators $\operatorname{Cor}^{\catname C}_A(\Sigma)$ are mapping class group invariant \eqref{eq:Cor-mod-inv} and so is the mapping class group average $\langle x \rangle_\Sigma$ by definition and, therefore, they both lie in the same 1-dimensional space. For the statement to hold, we only need to show that $\langle x \rangle_\Sigma$ is non-zero. To show that this is the case, it is sufficient to check that $\pairing_\Sigma(\langle x\rangle_\Sigma)\neq 0$. This follows from the invariance property \eqref{eq:pairing-inv} and linearity of $\pairing_\Sigma$, namely 
\be
\pairing_\Sigma(\langle x\rangle_\Sigma) = \frac{1}{|G|} \sum_{g\in G}{\pairing_\Sigma(g.x)} = \frac{1}{|G|}\sum_{g\in G}{\pairing_\Sigma (x)} = \pairing_\Sigma(x) \neq 0~.  
\ee
\end{proof}

The element $x$ in Theorem~\ref{thm:mcgavg-cor} is referred to as the \textit{seed}. 
Consider a handlebody $H_\Sigma$ with boundary $\Sigma$ without framed points, seen as a bordism $H_\Sigma:\emptyset \rightarrow \Sigma$. The RT TQFT evaluated on its double 
offers a natural choice of a seed
\begin{equation}\label{eq:vacuum-handlebody}
	x_\mathrm{vac} := \mathcal{Z}^\mathcal{C}(H_\Sigma\sqcup -H_\Sigma)
\end{equation}
called \textit{vacuum seed}. 
Notice that $x_\mathrm{vac}$ satisfies $\pairing_\Sigma(x_\mathrm{vac}) \neq 0$ which follows directly from Example~\ref{ex:handlebody-pairing} (where $x_\mathrm{vac} = x \otimes \hat{x}$). 
As we will see in Section~\ref{sec:physics-review} $x_\mathrm{vac}$ will be the vacuum contribution to the gravity path integral. For the torus $\Sigma = T^2$ the element $\mathcal{Z}^\mathcal{C}(H_{T^2})$ corresponds to the vacuum character $\chi_0(\tau)$ \cite[Eq.~(5.15),(5.16)]{FRS1} ($\tau$ is the conformal parameter) and $x_\mathrm{vac}$ to $|\chi_0(\tau)|^2$.

The correspondence \eqref{eq:mcgavg-cor} has been established for the vacuum seed $x_\mathrm{vac}$ and the Ising CFT on the torus in \cite{Castro:2011zq} and then extended to partition functions on higher genus surfaces in \cite{Jian:2019ubz}. 
Theorem~\ref{thm:mcgavg-cor} extends these results to arbitrary RCFTs satisfying irreducibility and property F, and to more general seeds $x$ (i.e.\ all $x$ s.th.\ $\pairing_{\Sigma}(x)\neq 0$). In Lemma~\ref{lem:tangle-pairing} we described a large class of possible seeds $x$ in terms of tangles. 

\begin{remark}\label{rem:thm-mcg-avg-cor}
    Theorem \ref{thm:mcgavg-cor} gives a correspondence between conformal correlators and mapping class group averages for a fixed choice of d-surface $\Sigma$. Suppose now that the MFC $\catname C$ satisfies the assumptions in Theorem \ref{thm:mcgavg-cor} for all d-surfaces. Then we can obtain any conformal correlator via a mapping class group average (up to a non-zero scalar). This is less surprising in view of the main result in \cite{RR}:
    \begin{quote}
    Let $\catname C$ be an MFC such that the mapping class group representations $V^{\catname C}(\Sigma)$ are irreducible for any surface $\Sigma$ (with no insertions). Then there exists a unique Morita class of simple symmetric special Frobenius algebras, namely that of the trivial algebra $\unit$.
    \end{quote}
The above statement implies in particular that for such $\catname C$ there is a unique (up to equivalence) consistent collection of correlators $\operatorname{Cor}^{\catname C}(\Sigma)$, and it is these correlators which are produced by the mapping class group averages.
\end{remark}

\section{Irreducibility and property F}
\label{sec:IrredPropF}
In this section we discuss theories with the irreducibility 
    and
finiteness property used in the hypothesis of the main correspondence Theorem~\ref{thm:mcgavg-cor}. We review in Section~\ref{subsec:irred-propF-review} some known examples and in Section~\ref{subsec:irred-Ising} 
we show that Ising categories satisfy both properties, extending results from \cite{Jian:2019ubz} to surfaces with framed points. 

\subsection{Irreducibility and property F examples}\label{subsec:irred-propF-review}

    \subsubsection{Irreducibility on surfaces with and without framed points}\label{sec:irred-points-or-not}

A trivial example which has irreducible representations and property F with respect to any d-surface $\Sigma $ is $\Vect$. There are further trivial examples of theories with irreducible representations when restricting to specific surfaces. For instance, any MFC $\catname{C}$ has an irreducible representation on the sphere, as $V^\mathcal{C}(\mathbb{S}^2)\cong \mathbbm{k}$.
However, we are interested in examples which satisfy the irreducibility property for a large family of surfaces.

\medskip
 
The only non-trivial examples with irreducible representations $V^{\catname C}(\Sigma)$'s for all surfaces $\Sigma$ 
we are aware of 
are Ising-type categories and the MFC $\catname C(\mathfrak{sl}_2,k)$ associated to the affine Lie algebra $\widehat{\mathfrak{sl}}_2$ at certain levels $k \in \mathbb{Z}_{>0}$. 
Let us list these examples, as well as some non-examples. 
\begin{enumerate}
	\item 
	It is shown in \cite{Roberts:1999} that for $\catname C = \catname C(\mathfrak{sl}_2,k)$ and $r = k+2$ prime,
	all projective representations $V^{\catname C}(\Sigma_{g,0})$, $g \ge 0$ are irreducible.
	
	Most of the remaining cases can be excluded already by looking at $g=1$. Namely by \cite[App.\,A]{Gepner:1986hr} and \cite[Prop.\,1]{Cappelli:1987xt}, invariants in the representation $V^{\doublecat}_{g=1}$ are obtained from divisors $d$ of $r$, with divisors $d$ and $r/d$ describing the same invariant subspace, and where $d$ with $d^2=r$ is excluded.
	Thus, when $r \ge 3$ is not a prime or a square of a prime,
	the space of invariants satisfies
	$\dim (V^{\doublecat}_{1})^{\Mod_{1}}>1$, and so by Lemma~\ref{lem:mcg-invariant-subsp}, $V^{\catname C}_{1}$ is not irreducible. 
	
	On the other hand, for $k=2$ ($r = 4$), one obtains a category of Ising-type, for which all $V^{\catname C}_{g}$ are irreducible, see point 2. 
	Some results on the irreducibility of $V^{\catname C}_{g \ge 2}$ for the remaining cases of $r = p^2$ with $p>2$ prime can be found in \cite{Korinman} but this remains an open problem and we hope to return to this in the future.
	
	\item 
	The Ising model without framed points is studied in \cite{Castro:2011zq,Jian:2019ubz}. Irreducibility of all $V^{\catname C}_g$, $g \ge 0$ is shown in \cite[Sec.\,4.3]{Jian:2019ubz}. In the next section we will extend this result to all 16 Ising-type MFCs and to surfaces with framed points. 
	
	\item 
	Let $\catname C = \catname C(\mathfrak{sl}_N,k)$ be the MFC for the affine Lie algebra $\widehat{\mathfrak{sl}}_N$ at any
	level $k \in \mathbb{Z}_{>0}$, for $N \ge 3$. It is shown in \cite[Thm.\,3.6]{Andersen:2009} that the $V^{\catname C}_g$ are reducible for each $g \ge 1$. 
	
	\item 
	For $\catname C(\mathfrak{sl}_2,k)$, irreducibility has also been studied for the mapping class group of surfaces with framed points in \cite{KoSa} where it is shown with respect to surfaces which contain at least one framed point labelled by the fundamental label 1. Irreducibility for surfaces with boundary has been studied in \cite{KuMi}. For Ising-type MFCs, irreducibility in the presence of framed points will be shown in the next section.
\end{enumerate}

\subsubsection{Property F on higher genus surfaces}

Similar to irreducibility, there are trivial examples where property F is present. For example, according to the Remark~\ref{rem:propF}, if the unframed pure mapping class group $\PMod_{g,n}^\mathrm{un}$ is trivial (i.e.\ for $(g,n)=(0,n)$ with $n\leq 3$) then $\PMod_{g,n}$ is fully generated by Dehn twist around its framed points and its representation image is finite. 

A non-trivial result of \cite{Ng}, briefly mentioned in Remark~\ref{rem:thm-mcg-avg-cor}, guarantees the presence of property F on the torus $\Sigma = T^2$ for any MFC $\catname{C}$. However, we are also interested in examples on a larger family of surfaces, particularly of higher genus. Such examples include:

\begin{enumerate}
	\item The MFC $\catname{C}=\operatorname{Rep}(D^\omega G)$ of representations of the twisted Drinfeld double of a finite group $G$ is shown in \cite{Gustaf} to have property F with respect to any d-surface (any genus and number of framed points). 
	\item  In \cite{Jian:2019ubz} it is shown that the MFC associated to the Ising CFT has property F with respect to surfaces without framed points. In the next section we extend this to all 16 Ising MFCs and all d-surfaces. 
\end{enumerate}

\subsubsection{Property F on genus $\boldsymbol{0}$ surfaces}
\label{subsubsec:propFgenus0}

As mentioned before, the term property F is borrowed from \cite{NR} where such a finiteness property is studied with respect to braid group representations. 

Let $\catname{B}$ be a braided fusion category. For an object $X \in \catname{B}$, the endomorphism space $\End(X^{\otimes n})$ carries a natural action of the braid group by mapping a braid generator in $B_n$ to the corresponding braiding of the strands labelled by $X$. The braided fusion category is said to have property $\mathrm{F}_\mathrm{braid}$ with respect to $X$ if the braid group image is finite for all $n$. If this is the case for every $X$ and any $n\in \mathbb{N}$, then $\mathcal{B}$ is said to have property $\mathrm{F}_\mathrm{braid}$ as a braided fusion category. 

\begin{remark}
	Since braid groups can be seen as mapping class groups as in Example~\ref{ex:braid-group}, we may compare the above mentioned notion of finiteness with that of the mapping class group. 
	
	Let $\catname{C}$ be a MFC and $X$ an object in $\catname{C}$. Let $\mathbb{S}^2(X,\dots, X,i)$ be the sphere with the first $n$ points labelled by $X$ and the $(n+1)$'th point labelled by $i\in I$. By semisimplicity, property $\mathrm{F}_\mathrm{braid}$ with respect to $X$ (as a braided fusion category) is equivalent to property F being valid for $\mathbb{S}^2(X,\dots, X,i)$ for all choices $i\in I$. 
We are not aware of a counterexample MFC which satisfies property $\mathrm{F}_\mathrm{braid}$ but not property F.   
\end{remark}
   
Before discussing results on property $\mathrm{F}_\mathrm{braid}$, we recall the following notions.   
   
A fusion category $\mathcal{A}$ is \textit{weakly integral} if its Frobenius-Perron dimension\footnote{
The Frobenius-Perron dimensions are determined by the unique character on the fusion ring of $\mathcal{A}$ such that $\operatorname{FPdim}(i) > 0 $ for all $i\in I$, see \cite[Ch.\ 3.3]{EGNO} for more details. 
}
is an integer, \ie $\mathrm{FPdim}(\catname{A})\in \mathbb{N}$. This is equivalent to the condition that $\operatorname{FPdim}(X)^2 \in \mathbb{N}$ for all simple objects $X$ \cite[Prop.~8.27]{ENO}. If $\operatorname{FPdim}(X)\in \mathbb{N}$ for all simple objects $X$, then $\mathcal{A}$ is \textit{integral}.
Ising categories are examples of weakly integral (but non-integral) fusion categories, since $\operatorname{FPdim}(\sigma) = \sqrt{2}$ for the spin object $\sigma \in I$. Ising categories will be introduced in the next section in detail. 

A fusion category is \textit{group-theoretical} if it is Morita equivalent to a pointed fusion category, i.e.\ equivalent to $\Vect_G^\omega$ for some group $G$. Equivalently, a group-theoretical fusion category $\mathcal{A}$ has Drinfeld centre $\mathcal{Z}(\mathcal{A})\simeq \operatorname{Rep}(D^\omega G)$. Every group-theoretical fusion category is integral, but the converse does not hold \cite{Nikshych08}.

There is a further notion of \textit{weakly group-theoretical} fusion categories, which by definition means that they are Morita equivalent to
nilpotent\footnote{A fusion category $\mathcal{N}$ is \textit{nilpotent} if the sequence $\mathcal{N}\supset \mathcal{N}_\mathrm{ad} \supset (\mathcal{N}_\mathrm{ad})_\mathrm{ad}\supset \dots$, constructed by taking adjoint subcategories, converges to $\Vect$. 
} 
fusion categories. Weakly group-theoretical fusion categories are weakly integral \cite{ENO2}. The converse statement is an open question. 

\medskip

Regarding property $\mathrm{F}_\mathrm{braid}$, it is conjectured by the authors in \cite[Sec.\,1]{NR} that
a braided fusion category $\mathcal{B}$ has property $\mathrm{F}_\mathrm{braid}$ if and only if it is weakly integral.
This conjecture has been verified partially in one direction by showing that any weakly group-theoretical braided fusion category has property $\mathrm{F}_\mathrm{braid}$ \cite{GrN}. 
This extends an earlier result in \cite{ERW} that group-theoretical braided fusion categories have property $\mathrm{F}_\mathrm{braid}$ as well as some weakly group-theoretical examples \cite{RW,GRR}. 

\begin{remark}
	\begin{enumerate}
	\item
	One motivation for studying braid group representations and particularly their finiteness or lack thereof is due to topological quantum computing. To have universal quantum computation one is interested in dense (thus infinite) braid group images. A famous example is that of the Fibonacci category $\operatorname{Fib}$. There are many non weakly integral examples coming from $\catname{ C}(\mathfrak{g},k)$. For instance, $\catname{C}(\mathfrak{sl}_2,k)$ is weakly integral only if  $k\in \{2,3,4,6\}$ \cite{NR}. If the conjecture in \cite[Sec.\,1]{NR} holds, this also describes all cases where $\catname{C}(\mathfrak{sl}_2,k)$ has property F.

 \item From part 1 one can observe that there are many levels which admit irreducible mapping class group representations but not $\mathrm{F}_\mathrm{braid}$, \eg for all $k\geq 5$ with $k+2$ prime. Conversely, property F does not imply irreducibility (\eg levels $k=4,6$).
 \end{enumerate}
\end{remark}

\subsection{Irreducibility and property F of Ising categories}\label{subsec:irred-Ising}
In this section, we will review Ising-type MFCs and compute their associated mapping class group representations explicitly. Thereafter, we will present a proof that such categories satisfy property F and the irreducibility property, thus extending the result of \cite{Jian:2019ubz} to surfaces with framed
points and all 16 Ising-type MFCs.

\subsubsection{Ising categories}
\label{subsec:isingcats}
We introduce Ising categories following \cite[App.\,B]{DGNO} and compute their mapping class group action explicitly following the previous section. Up to equivalence there are 16 Ising modular fusion categories. This family is parameterised by an 8'th root of $-1$, which will be denoted by $\zeta$, and a sign $\nu \in \{\pm\}$. The first determines the braided structure of the category, whereas the second determines the spherical structure. 
We will also need
\be \lambda := \zeta^2 + \zeta^{-2} \ ,
\ee
which satisfies $\lambda^2 = 2$. 

The Ising category $\catname{C}(\zeta, \nu)$ has three simple objects $\unit, \varepsilon, \sigma$ and the fusion rules are given in the following table:
\begin{equation}
\begin{tabular}{c|c c c}        $\otimes$ & $\unit$ & $\varepsilon$ & $\sigma$ \\ \hline
         $\unit$ & $\unit$ & $\varepsilon$ & $\sigma$ \\
         $\varepsilon$ & $\varepsilon$ & $\unit$ & $\sigma$ \\
         $\sigma$ & $\sigma$ & $\sigma$ & $\unit \oplus \varepsilon$
    \end{tabular}
\end{equation}    
In particular, there are no multiplicities in the fusion basis, i.e.\ $N_{ij}^k \in \{0,1\}$. This simplifies the notation of $R$,$F$,$B$-matrices and the associated fusion graphs. Recall that $F$-matrices encode associativity of the monoidal structure, $R$-matrices econde the braided structure and $B$-matrices are a combination of both. For a detailed introduction of $R$,$F$,$B$-matrices see Appendix~\ref{app:rfb-matrices}. 
The dimensions are $d_\unit = d_{\varepsilon} = 1 $, $d_\sigma = \nu \lambda $, and the twist values are $\theta_\unit = 1$, $\theta_\varepsilon = -1$ and $\theta_\sigma  = \nu \zeta^{-1} $. 

The braidings in terms of $R$-matrices are given by: 
\be
R^{(\varepsilon \varepsilon)\unit} = -1
~,\quad  
R^{(\varepsilon \sigma)\sigma}= R^{(\sigma \varepsilon)\sigma} = \zeta^4
~, \quad 
R^{(\sigma \sigma)\unit} = \zeta
~,\quad 
R^{(\sigma \sigma)\varepsilon} = \zeta^{-3} \ .
\ee
The (non-trivial) F-matrices are given as: 
\begin{align}
    F^{(\varepsilon \sigma \varepsilon) \sigma} &= F^{(\sigma \varepsilon \sigma) \varepsilon} = -1 \ ,
    \nonumber
\\
F^{(\sigma \sigma \sigma) \sigma} &= \begin{pmatrix} F^{(\sigma \sigma \sigma) \sigma}_{\unit \unit} & F^{(\sigma \sigma \sigma) \sigma}_{\unit \varepsilon} \\ F^{(\sigma \sigma \sigma) \sigma}_{\varepsilon \unit} & F^{(\sigma \sigma \sigma) \sigma}_{\sigma\sigma} \end{pmatrix}= \frac{1}{\lambda} \begin{pmatrix} 1 & 1 \\ 1 & -1 \end{pmatrix}
~.
\end{align}
The $S$-matrix is
\be\label{eq:IsingSmatrix}
S =\frac{1}{2} \begin{pmatrix} 1& 1& \nu \lambda\\
1 & 1 & -\nu \lambda \\
\nu\lambda & -\nu\lambda & 0 \end{pmatrix}
~~,
\ee
where the order of the basis is $\{\unit, \varepsilon,\sigma\}$. 

Notice that the F-matrices are self-inverse, \ie $F^{(ijk)l} =G^{(ijk)l}$. 
It will be convenient to use the following notation: $\{\varepsilon_a \}_{a\in \mathbbm{Z}_2} $, where $\varepsilon_{0}= \unit $ and $\varepsilon_{1} = \varepsilon$. With this notation we have for example: 
$R^{(\sigma \sigma)\varepsilon_a} = \zeta^{1-4a}$.

Using the above data one can also compute the B-matrices according to \eqref{eq:B-matrix-explicit} and \eqref{eq:B-matrix-rev-explicit}. When $\Hom(l,i\otimes j \otimes k)$ is 1-dimensional (in other words when at least one of the labels is not $\sigma$) we have the following non-trivial B-matrices: 
\begin{align}\label{eq:Ising-B-matrix-1}
	B^{\pm(\varepsilon_a \varepsilon \varepsilon)\varepsilon_a} &= (-1)^a ~,
 \nonumber\\
        B^{\pm(\varepsilon_a \sigma \varepsilon_b)\sigma} &= B^{\pm(\sigma \varepsilon_b \sigma)\varepsilon_a} = B^{\pm(\sigma \sigma \varepsilon_b)\varepsilon_a} = B^{\pm(\varepsilon_a \varepsilon_b \sigma)\sigma} = (-1)^{ab} \zeta^{\pm4b} ~,
        \nonumber\\
	B^{\pm(\varepsilon_a \sigma \sigma)\varepsilon_b} &= \zeta^{\pm(1- 4[a+b])}~.
\end{align} 
Here, $[a+b] \in \{0,1\}$ is defined to be equal to $a+b \mod 2$. 
When all labels $i,j,k$ and $l$ are $\sigma$-labels the corresponding $\Hom$-space is 2-dimensional. By inserting R- and F-matrices into \eqref{eq:B-matrix-explicit} and \eqref{eq:B-matrix-rev-explicit} one obtains the following $2\times2$-matrices: 
\begin{equation}\label{eq:Ising-B-matrix-2}
	B^{\pm(\sigma \sigma \sigma)\sigma} = \frac{\zeta^{\pm1}}{\lambda}
 \begin{pmatrix}
		\zeta^{\mp2} & \zeta^{\pm2} \\
		\zeta^{\pm2} & \zeta^{\mp2}
	\end{pmatrix}~.
\end{equation}
\begin{remark}
	Half of the 16 Ising categories are unitary and the other half are non-unitary dictated by the positivity respectively negativity of $d_\sigma$.	
	The Ising category associated to the Ising CFT as well as the $\catname{C}(\mathfrak{sl}_2,k)$ at level $k=2$ are among the unitary Ising categories.
\end{remark}

\subsubsection{Mappling class group action of Ising categories}
\label{subsec:ising-mcg-action}

In this section, we will explicitly describe how the pure mapping group acts on the associated state spaces in the case of Ising categories. Therefore, we fix an Ising category $\catname{C} = \catname{C}(\zeta,\nu)$ as introduced in 
Section~\ref{subsec:isingcats}. 

Let $\Sigma$ be a surface of genus $g$ and $n$ framed points labelled by simple objects $l_1,\dots, l_n$, \ie the underlying vector space of the state space \eqref{eq:V_g,n-def} is
\begin{equation}\label{eq:Ising-V_gn}
 V_{g,n} = \catname C(\unit, l_1\otimes \dots l_n \otimes L^{\otimes g}) ~.   
\end{equation}
We now describe a basis of this space that we will use to give the action of the generators as listed in \eqref{eq:MCG-gen-action}.
Consider the oriented graph  
\vspace{2ex}
\be \label{eq:basisgraph}
\Gamma :=
    \boxpic{1}{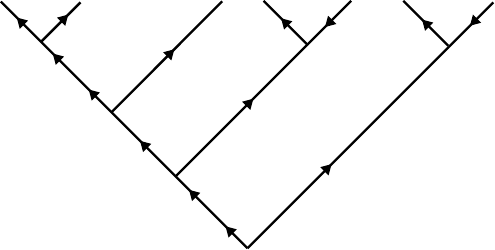}{
        \put (-3,51) {$L_1$}
        \put (15,51) {$L_2$}
        \put (42,51) {$L_n$}
        \put (52,51) {$J_1$}
        \put (70,51) {$J'_1$}
        \put (81,51) {$J_g$}
        \put (99,51) {$J'_g$}
        \put (6,35) {$D_3$}
        \put (52,25) {$K_1$}
        \put (67,12) {$K_g$}
        \put (22,15) {$D_{n+1}$}
        \put (40,-1) {$A_g$}
        \put (58,19) {$\ddots$}
        \put (15,25) {$\ddots$}
        \put (26,51) {$\dots$}
    }
\vspace{2ex}
\ee
where $L_1,\dots, L_n, D_3,\dots, D_{n+1}, K_1,\dots, K_{g}, A_2,\dots, A_{g}, J_1,\dots, J_{g}, J'_1, \dots , J'_{g}$ are the edges of the graph. This choice of notation will be convenient later. 

\medskip

A colouring is a map $\chi: E(\Gamma)\rightarrow I$, which assigns to every edge of the graph a label in $I$.
We say that a colouring is \textit{admissible} if at any trivalent vertex of the graph the associated fusion rule is non-zero. For instance, the colouring $\chi$ is admissible at the first vertex in \eqref{eq:basisgraph} if $N_{\chi(L_1)\chi(L_2)}^{\chi(D_3)}\neq 0$. 
Let $\colG$
denote the set of admissible colourings of the graph $\Gamma$. 
Moreover, we consider the subset $\colGo \subset \colG$ consisting of all colourings $\chi$ such that 
\begin{equation}
\chi(L_1)= l_1 ~,~ \dots ~,~ \chi(L_n)= l_n
\quad \text{and} \quad  \chi(J_k)=\chi(J'_k)~~,~~k= 1,\dots, g~.
\end{equation}
 Notice that for $\chi\in \colGo$ we have $\chi(K_m) \in \{\unit, \varepsilon\}$ for all $m=1, \dots, g$ as $\chi(K_m)= \sigma$ is not admissible due to the condition of $\chi(J_m)= \chi(J'_m)$ and the Ising fusion rules. 

In terms of this notation, for $\chi\in \colGo$ the graph $\Gamma$ coloured by $\chi$ represents a vector 
$\hat{\chi}~\in~V_{g,n}$,
with $V_{g,n}$ as in \eqref{eq:Ising-V_gn}, and where the embedding $\chi(J_k)\otimes \chi(J_k)^\ast \to L$ is implicit.
Altogether, the set 
\begin{equation}\label{eq:V_gn-basis}
\{ \hat\chi \}_{\chi \in \colGo}
\end{equation}
forms a basis of $V_{g,n}$.

We proceed by computing the mapping class group action on $V(\Sigma_{g,n})$ with respect to this basis using the equations in \eqref{eq:MCG-gen-action}. It is clear that 
\be
T_{\lambda_m} \hat{\chi} = \theta_{l_m}\, \hat{\chi}
\ee
and 
\be\label{eq:alpha-action}
T_{\alpha_m} \hat{\chi} = \theta_{\chi(J_m)}\,\hat{\chi}~.
\ee
To give the expression for the action of $S_m$ on a basis vector $\hat{\chi}$ we make the following distinction: 

For $\chi(J_m)= \varepsilon_p$, which by admissibility also implies $\chi(K_m) = \unit$, we compute
\be\label{eq:Saction-1}
S_m \hat\chi = \frac{1}{2}\left( \hat{\chi}|_{J_m=\unit} + \hat{\chi}|_{J_m=\varepsilon} + (-1)^p \nu\lambda\, \hat{\chi}|_{J_m=\sigma}\right) \ ,
\ee
where the slash notation indicates that for example $\hat{\chi}|_{J_m= \sigma}$ represents the vector with the same colouring $\chi$ everywhere up to the edge $J_m$, which is now labelled by $\sigma$. The factor $\frac{1}{2}$ appears as $D=2$
for Ising categories.

For $\chi(J_m)=\sigma$ and $\chi(K_m)=\unit$ we find
\be\label{eq:Saction-sigma}
S_m \hat{\chi} = \frac{\lambda}{2}\left( \hat{\chi}|_{J_m=\unit} - \hat{\chi}|_{J_m=\varepsilon}\right)~. 
\ee

 Equations \eqref{eq:Saction-1} and \eqref{eq:Saction-sigma} can be formed into one: 
\be 
S_m \hat\chi = \sum_{j\in I}{S_{\chi(J_m),j}\,\hat{\chi}|_{J_m = j}}
\ee
using the $S$-matrix $S$ from \eqref{eq:IsingSmatrix} where $\chi$ such that $\chi(K_m) = \unit$.

The last case is when $\chi(K_m)=\varepsilon$ and hence $\chi(J_m) = \sigma$. The result is 
\be\label{eq:Saction-sigma2}
S_m \hat\chi = \zeta^2 \,\hat\chi
\ee
which follows from the computation (see \cite[Eq.\ (6.16)]{Romaidis-Thesis})\vspace{1ex}
\begin{equation}\label{eq:sactioncomp}
\boxpic{1}{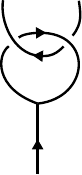}{
    \put (18,-13) {$k$}
\put (1,103) {$\sigma$}
\put (40,103) {$\sigma$}
\put (-3,40) {$\sigma$}
}~
	=~\delta_{k,\varepsilon}~\frac{2\zeta^2}{d_\sigma}~
	\boxpic{1}{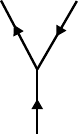}{
	 \put (20,-13) {$k$}
	\put (-1,103) {$\sigma$}
	\put (47,103) {$\sigma$}
}~.
	\vspace{2ex}
\end{equation}

The action of $T_{\gamma_m}$ is described by twisting with $\theta$ the tensor product of two neighbouring strands as seen in equation \eqref{eq:MCG-gen-action}. If the fusion is again a simple object, the action produces just the twist eigenvalue of this simple object. That is, let $\chi \in \colGo$ such that $\chi(J_m)\otimes \chi(J_{m+1})$ is simple. Then, we have 
\be\label{eq:Tgamma-nonsigma}
T_{\gamma_m} \hat\chi = \theta_{\chi(J_m)\otimes\chi(J_{m+1})}\,\hat{\chi}~. 
\ee
If $\chi(J_m)\otimes \chi(J_{m+1})$ is not simple, then we have $\chi(J_m)= \chi(J_{m+1}) = \sigma$ according to the Ising fusion rules. In this case, we find \cite[Eq.\ (6.20)]{Romaidis-Thesis}
\be\label{eq:Tgamma-sigma}
T_{\gamma_m} \hat\chi = \hat{\chi}|_{\{K_m,K_{m+1},A_{m+1}\}\otimes \varepsilon} \ ,
\ee 
where $\{K_m,K_{m+1},A_{m+1}\}\otimes \varepsilon$ indicates that we change the colouring of the edges $K_m, K_{m+1}$ and $A_{m+1}$ to $\chi(K_m)\otimes \varepsilon, \chi(K_{m+1})\otimes \varepsilon$ and $\chi(A_{m+1})\otimes \varepsilon$ respectively. 
 
To compute the action of $T_{\delta_m}$ we make use of the tensoriality property of the twist. 
For instance, the action of $T_{\delta_m}$ on $\hat{\chi}$ is the composition of pure braids in strands labelled by $\chi(L_{m+1}),\dots, \chi(L_n), \chi(J_1)$ and products of the twist eigenvalues of these labels. Therefore, we study the action of pure braids.  

We start by considering pure braids in $l_1,\dots, l_n$. Let
\vspace{2ex} 
\be
T_{ij}~ \hat{\chi}~=\quad \boxpic{1}{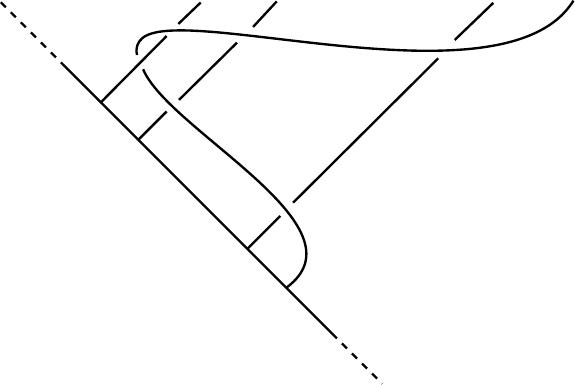}{
\put (3,48) {$\chi(D_i)$}
\put (8,40) {$\chi(D_{i+1})$}
\put (35,17) {$\chi(D_j)$}
\put (42,6) {$\chi(D_{j+1})$}
\put (34,68) {$l_i$}
\put (47,68) {$l_{i+1}$}
\put (82,68) {$l_{j-1}$}
\put (99,68) {$l_j$}
\put (25,27) {$\ddots$}
\put (63,68) {$\dots$}
}\ee
be the pure braid between the $l_i$- and $l_j$-coloured strands. 

If $l_i$ and $l_j$ have a unique fusion, \ie if $l_i\otimes l_j$ is a simple object, then $T_{ij}$ acts by a phase $\theta_{l_i\otimes l_j}\theta^{-1}_{l_i} \theta^{-1}_{l_j}$, which is obtained by expressing the double braid using the twist and its inverses. 

Otherwise, for $l_i = l_j = \sigma$, we prove the following lemma which proves useful later. 

\begin{lemma}\label{lem:purebraid}
If $l_i = l_j =\sigma$, then \[ T_{ij} \hat{\chi} = \zeta^2 (-1)^{p(\chi(D_{j+1}))} \zeta^{4m}\, \hat{\chi}|_{\{D_{i+1},\dots, D_{j}\}\otimes \varepsilon} \ ,
\]
where $p(\sigma)=0$, $p(\varepsilon_a) = a$ 
and some $m\in \mathbb{Z}_4$ which may depend on $\chi(D_i),\dots, \chi(D_{j})$ but it does not depend on $\chi(D_{j+1})$.
\end{lemma}
\begin{proof}
We prove this by induction on $|i-j|$. We start with $j = i +1 $. This is the case of applying the $B$-matrix twice on the strands $l_i$ and $l_{i+1}$. We compute this for the different admissible labels of $D_{i+1} = D_{j}$. 
\begin{itemize}
    \item For $\chi(D_{i+1}) = \sigma$, let $\chi(D_{i}) = \varepsilon_a$ and $\chi(D_{i+2}) = \varepsilon_{b}$ be the admissible labels. Then, we find: 
    \[ T_{ii+1} \hat{\chi} = (B^{(\varepsilon_a \sigma \sigma) \varepsilon_b})^2 \,\hat{\chi} \overset{\eqref{eq:Ising-B-matrix-1}}{=} \zeta^2 (-1)^{a+b} \,\hat{\chi} \ .
    \]
    \item For $\chi(D_{i+1}) = \varepsilon_a$, we get admissible labels $\chi(D_i) =\sigma$ and $\chi(D_{i+2}) = \sigma$. The result is:
    \[ T_{ii+1} \hat{\chi} \overset{\eqref{eq:Ising-B-matrix-2}}{=} \zeta^2\, \hat{\chi}|_{D_{i+1}\otimes \varepsilon}~.\]
\end{itemize}

Assuming that the statement is true for fixed $i$ and $j$, we make the induction step and prove it for $i$ and $j+1$:
\begin{align*}
        T_{ij+1} \hat{\chi} =\zeta^2 \zeta^{4 m} \sum_{s_{j+1},d_{j+1}'}{(-1)^{p(s_{j+1})}\, B^{-(d_j l_{j+1} l_j)d_{j+2}}_{d_{j+1}s_{j+1}}\, B^{(\varepsilon \otimes d_j l_{j} l_{j+1})d_{j+2}}_{s_{j+1}d'_{j+1}}} \\
        \hat{\chi}|_{\{D_{i+1},\dots,D_{j}\}\otimes \varepsilon, D_{j+1}= d'_{j+1}}.
\end{align*}
where $m$ depends only on $d_{i} \equiv \chi(D_i),\dots, \chi(D_j)\equiv d_{j}$ allowing us to take the factor $\zeta^{4m}$ out of the sum. 
This is the whole reason why we assume this special dependence of $m$ on the D-labels; for the induction to work. 
The above is obtained by using the inverse $B$-matrix, the induction assumption and the $B$-matrix again.
Recall that there is a pure braid relation 
\[
T_{i j+1} = (\id{}\otimes c^{-1}_{j-1,j})\circ T_{ij}\circ (\id{}\otimes c_{j-1,j}) \ ,
\] where $T_{ij}$ denotes the pure braid of strands $i$ and $j$ as before and $c$ 
denotes the braid generator with $c^{-1}$ the inverse braid.
After computing this for all labels, one verifies the statement. We give the phases for each case: 

\begin{itemize}
    \item For $l_{j} = \varepsilon$: $\zeta^2 \zeta^{4m} (-1)^{p(d_{j+2})+1}$.
    \item For $l_{j} = \sigma$: \begin{enumerate}
        \item for $d_{j+2} = \varepsilon_{a}$: $\zeta^2 \zeta^{4m+2p(d_j)}(-1)^a$
        \item for $d_{j+2} = \sigma$: $\zeta^2 \zeta^{4m} \zeta^4 (-1)^{p(d_{j+1})}$
    \end{enumerate}
\end{itemize}
\end{proof}

Similar to this, we can prove the following lemma, which includes braids with $i_{1}$.

\begin{lemma}\label{lem:purebraid2}
If $l_i = i_1 =\sigma$, then \[ T_{l_ii_1} ~\hat\chi ~=~ \zeta^2 \zeta^{4m}\, \hat{\chi}|_{\{D_{i+1},\dots, D_{n+1},K_1\}\otimes \varepsilon}
\]
with $p$ as in Lemma~\ref{lem:purebraid}
and $m\in \mathbb{Z}_4$ which may depend on $\chi(D_i),\dots, \chi(D_{n+1})$ but not on $\chi(A_2)$.
\end{lemma}

\subsubsection{Irreducibility property of Ising categories}
\label{subsec:irred}

In this section we prove our second main result:

\begin{theorem}\label{thm:Ising-irred}
Let $\catname C = \catname C(\zeta,\nu)$ be an Ising-type MFC and let $\Sigma_{g,n}$ be a d-surface whose framed points are labelled by simple objects. Then, $V^{\catname C}(\Sigma_{g,n})$ is an irreducible projective representation of the pure mapping class group.
\end{theorem}

\begin{remark}
\begin{enumerate}
\item 
The case without framed points, i.e.\ $n=0$, was already shown in \cite[Sec.\,4.3]{Jian:2019ubz} (for one of the 16 Ising-type categories). It turns out that the same method works in all 16 cases and that it can be easily adapted to the case with framed points, and so our proof follows closely that of \cite{Jian:2019ubz}.

\item Irreducibility with respect to the pure mapping class group $\PMod_{g,n}$ implies irreducibility with respect to the mapping class group $\Mod^\mathcal{C}(\Sigma_{g,n})$, as $\PMod_{g,n}$ is a subgroup.

\end{enumerate}
\end{remark}

\begin{proof}[Proof of Theorem~\ref{thm:Ising-irred}]
Let $\Sigma_{g,n}$ be a d-surface 
with simple point labels and set $ V_{g,n} := V^{\catname C}(\Sigma_{g,n})$. To prove irreducibility, we will show that the inclusion $\mathbbm{k}[{V(\PMod_{g,n})}] \subset \End(V_{g,n}) $ is actually an equality. In terms of the basis $\{\hat\chi\}$ of $V_{g,n}$ from \eqref{eq:V_gn-basis}, denote by
$E(\hat\chi,\hat\chi')$ the elementary matrix in $\End(V_{g,n})$, which maps $\hat\chi$ to $\hat\chi'$ and which maps all other basis elements to zero. We will have completed the proof once we show that
\begin{equation}
    \forall \chi,\chi' \in \colGo ~:~ 
    E(\hat\chi,\hat\chi') \in \mathbbm{k}[{V(\PMod_{g,n})}] ~.
\end{equation}

As a first step towards this goal,
define the operators:  
\be
P_\unit (\gamma) = \frac{1}{16}\sum_{k=1}^{16}{V(T_{\gamma}^k)}
~, \quad 
P_{\sigma}(\gamma)  = \frac{(\unit - V(T_{\gamma}^2))}{1-\zeta^{-2}}
~,\quad
P_{\varepsilon}(\gamma) = \unit - P_{\unit}(\gamma) - P_{\sigma}(\gamma) 
~,
\ee
where $ \gamma  $ is a simple closed curve and $ T_{\gamma} $ is the Dehn twist around $\gamma$.   
Using $\frac{1}{16}\sum_{k=1}^{16} \theta_i^k = \delta_{\unit,i} $ and $\frac{1-\theta_i^2}{1-\zeta^{-2}} = \delta_{\sigma,i}$ one can check that 
that the $P_j(\gamma)$, $j \in \{\unit,\sigma,\varepsilon\}$,
for $\gamma$ a fixed simple closed curve,
are pairwise orthogonal idempotents. 
Note that for all simple closed curves $\gamma$, by construction 
\begin{equation}
P_\unit (\gamma)
\,,~
P_{\sigma}(\gamma)
\,,~
P_{\varepsilon}(\gamma) 
~\in~\mathbbm{k}[{V(\PMod_{g,n})}]~.
\end{equation}

\begin{figure}[t]
\centering
    \boxpic{0.7}{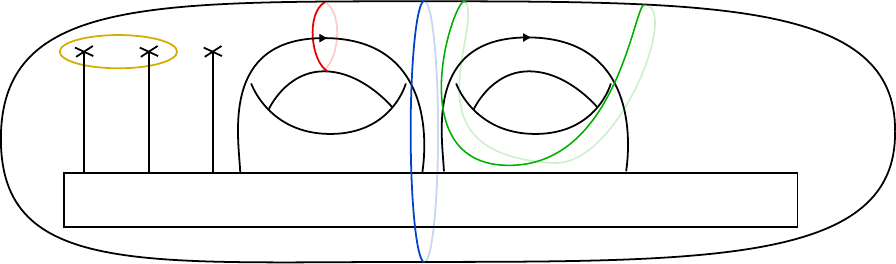}{
    	\put (34,30.5) {$J_1$}
    	\put (45,30.5) {$A_2$}
    	\put (77,18) {$\dots$}
    	\put (8,28) {$D_{i+1}$}
    	\put (6.5,15) {$l_1$}
    	\put (14,15) {$l_i$}
    	\put (20,15) {$l_{n}$}
    	\put (51,30.5) {$K_2$}
    }
    \caption{Curves of type $D,I,A,K$ on the surface based on the handlebody with an embedded graph and a coupon.}
    \label{fig:projcurves}
\end{figure}
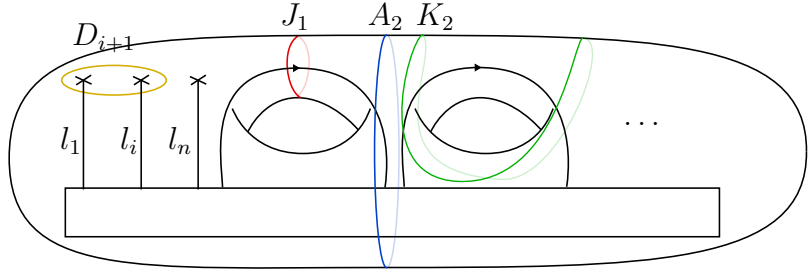

Let $a$ be one of the edge labels $D_m$, $J_m$, $A_m$, or $K_m$ as used in \eqref{eq:basisgraph} and denote by $\gamma_a$ the corresponding simple closed curve as shown in Figure \ref{fig:projcurves}.
The Dehn twist along $\gamma_a$ 
acts diagonally on each basis vector $\hat\chi$,
giving twist eigenvalue of the edge colour $\chi(a)$.
From this it is easy to check that the idempotents $P_j(\gamma_a)$ for $j\in I$ project onto those basis elements where the edge $a$ is labelled by $j$:
\begin{equation}
    P_j(\gamma_a)~ \hat{\chi}~ = ~\delta_{j,\chi(a)}~ \hat{\chi}
    ~.
\end{equation}
This implies that diagonal maps of the form $E(\hat \chi,\hat\chi) \equiv E(\hat \chi)$ are in $\mathbbm{k}[V(\Mod_{g,n})]$ as they are realised as a product of $P_{j}$ maps, namely 
\be
E(\hat{\chi}) = \prod_{e\text{ edge}}{P_{\chi(e)}(e)}~.
\ee 
More precisely, $e$ runs over the edges $D_3,\dots, D_{n+1},K_1,\dots, K_g, A_2,\dots,A_g, J_1,\dots, J_g$ and by abuse of notation they also denote the corresponding curves in Figure~\ref{fig:projcurves}. We will now describe how to move through different basis elements by changing the respective labels, \ie how to construct the rest of the $E(\hat\chi,\hat\chi')$ maps using mapping classes. We will use the computations of Section~\ref{subsec:ising-mcg-action}.

Changing only $J_m$ labels: Let $\chi\in \colGo$ be a colouring such that $\chi(K_m) = \unit$. Then, the colourings $\{\chi|_{J_m = j}\}_{j \in I}$ are all admissible. To jump between the corresponding basis vectors, we use \eqref{eq:Saction-1} and \eqref{eq:Saction-sigma} to show
\be\label{eq:Inothing-nothing}
E(\hat{\chi}|_{J_m= \varepsilon_p};\hat{\chi}|_{J_m= \varepsilon_{p+1}}) = 2 P_{\varepsilon_{p+1}}(J_m) \,S_m \,E(\hat{\chi}|_{J_m= \varepsilon_p})\ee 
and 
\be\label{eq:Inothing-sigma}
    E(\hat{\chi}|_{J_m= \varepsilon_p};\hat{\chi}|_{J_m= \sigma}) 
    = (-1)^p \lambda P_{\sigma}(J_m)\, S_m\, E(\hat{\chi}|_{J_m= \varepsilon_p})
\ee
and 
\be\label{eq:Isigma-nothing}
E(\hat{\chi}|_{J_m= \sigma};\hat{\chi}|_{J_m= \varepsilon_p})\\ = (-1)^p \lambda P_{\varepsilon_{p}}(J_m)\, S_m\, E(\hat{\chi}|_{J_m= \sigma}).
\ee
Note that $\chi(K_m) = \varepsilon$ only allows $\sigma$ as an $J_m$-label. 

\medskip

Changing $K_m$-labels: We will now describe how to change $K_m$-labels while keeping the labels of $D_3,\dots,D_{n+1}$ unchanged. Therefore, we will fix labels $d_3 \equiv\chi(D_3),\dots, d_{n+1}\equiv \chi(D_{n+1})$ and only consider basis elements with such labels. The labels of $A_2,\dots, A_g$ are in fact uniquely determined by the labels for $D_{n+1},K_1,\dots, K_g$ and they are all in the label set $\{\unit, \varepsilon\}$.

Consider first the case where $d_{n+1}= \unit$. In this case, $\chi(K_1)= \dots = \chi(K_g) = \unit $ and $\chi(J_1)= \dots = \chi(J_g) = \unit $ are admissible and let $\chi_0$ denote this colouring with the already fixed labels for $D_3,\dots, D_{n+1}$. 
For any other basis vector $\hat{\chi}$ (with the same $D$-labels $d_{n+1}= \unit $ as mentioned) we will construct the operators $E(\hat{\chi}_0, \hat\chi)$. Let $\{K_{m_l}\}_{l=1,\dots,L}$ be the maximal subset of $K$-edges such that $\chi(K_{m_l}) = \varepsilon$ and $m_1 <\dots <m_L$. In particular, we have $\chi(J_{m_l}) = \sigma$ for all $l$. The fusion rules imply that the number of non-trivial $K$-labels $L$ is even as we have fixed $d_{n+1}=\unit$. Then  
\be
E(\hat{\chi_0}, \hat{\chi}) ~=~E(\hat{\chi'},\hat{\chi})\, T_{\gamma_{m_{L-1},m_L}}\cdots T_{\gamma_{m_3,m_4}}\,T_{\gamma_{m_1,m_2}}\, E(\hat{\chi_0},\hat{\chi_0}|_{J_{m_l}=\sigma}) \ ,
\ee
where the curves $\gamma_{i,j}$
\begin{figure}[t]
\centering
\boxpic{0.7}{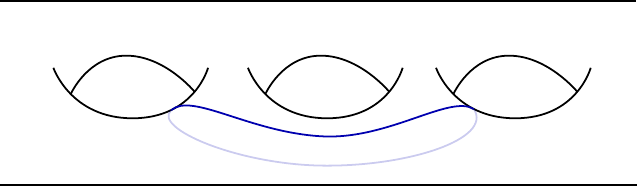}{
\put (30,8) {$\gamma_{ij}$}
}
\caption{The curve $\gamma_{ij}$ connects the $i$'th with the $j$'th genus.}
\label{fig:gammacurve}
\end{figure}
connect the $i$'th with the $j$'th hole as shown in Figure \ref{fig:gammacurve}. The action of $T_{\gamma_{m_i m_{i+1}}}$ is obtained similar to \eqref{eq:Tgamma-sigma} and changes the labels of $K_{m_i}$ and $K_{m_{i+1}}$ to $\varepsilon$. The operator $E(\hat{\chi_0},\hat{\chi_0}|_{\{J_{m_1},\dots,J_{m_L}\}=\sigma})$ changes the $J_{m_l}$ labels to $\chi(J_{m_l}) = \sigma$ and is a product of the operators in \eqref{eq:Inothing-sigma}. The vector $\hat\chi'$ is obtained from the colouring $\chi$, but with $\chi(J_m) = \unit$ for all $m\notin\{m_1,\dots,m_L\}$. Similarly, the operator $E(\hat{\chi'},\hat{\chi})$ is obtained as a product of operators in \eqref{eq:Inothing-nothing}, \eqref{eq:Inothing-sigma}.

In the opposite direction, one gets 
\be
E(\hat{\chi},\hat{\chi_0})~=~ E(\hat{\chi}',\hat{\chi_0})\, T_{\gamma_{m_{L-1},m_L}}\cdots T_{\gamma_{m_1,m_2}}\, E(\hat{\chi})~,
\ee
where $\hat{\chi}'$ now is obtained by $\chi$ but with all $K_m$-labels set to $\unit$.

Consider now the case where $d_{n+1} = \varepsilon$. In this case, fix the colouring $\chi_0$ with labels $\chi_0(K_1) = \varepsilon,\chi_0(J_1) = \sigma$ and $\chi_0(K_2)= \dots = \chi_0(K_g) = \unit$ and $\chi(J_2)= \dots = \chi(J_g)=\unit$. As before, let $\hat{\chi}$ be any basis element (now with $d_{n+1}= \varepsilon$) and consider the maximal subsequence $\{K_{m_l}\}_{l=1,\dots, L}$ with $\chi(K_{m_l}) = \varepsilon$, which is in this case has an odd number of elements $L$. Then, similar to the previous considerations we get 
\be
E(\hat{\chi_0}, \chi) ~=~E(\hat\chi', \chi)\, T_{\gamma_{m_{L-1},m_L}}\cdots T_{\gamma_{m_2,m_3}}\,T_{\gamma_{1,m_1}}\, E(\hat{\chi}_0, \hat{\chi}_0|_{J_{m_l}=\sigma}) \ ,
\ee
where we omit the $\gamma_{1,m_1}$ Dehn twist if $m_1=1$. Similarly,  
\be
E(\hat\chi,\hat{\chi_0})~=~ E(\hat{\chi}', \hat{\chi}_0)\, T_{\gamma_{m_{L-1},m_L}}\cdots T_{\gamma_{m_2,m_3}}\, T_{\gamma_{1,m_1}}\, E(\hat{\chi},\hat{\chi}|_{J_1 = \sigma})~.
\ee
Finally, note that the label $\sigma$ is not admissible for $d_{n+1}$ due to the Ising fusion rules, $\catname{C}(\sigma, L^g) = 0$ for Ising categories.

To conclude
we have constructed for any two basis elements $\hat\chi,\hat\chi'$ with fixed labels $d_3,\dots, d_{n+1}$ the operators $E(\hat\chi,\hat\chi')$, namely by passing through the distinguished basis element $\chi_0$, \ie
\be 
E(\hat\chi,\hat\chi')~=~ E(\hat{\chi}_0,\hat{\chi}')\, E(\hat{\chi},\hat{\chi}_0)~.
\ee
Changing $D_m$-labels: Changing the labels $D_m$-labels which previously were left unchanged, is the final step to complete the irreducibility proof. For $D_m$ to have two admissible labels, \ie such that both $\unit$ and $\varepsilon$ are admissible, we have either $\#(\sigma-\text{labelled framed points})\geq 4$ or $\#(\sigma-\text{labelled framed points})\geq 2$ and $g\geq 1$. Therefore, consider the case where $\unit$ and $\varepsilon$ are both admissible for $D_m$. There exists some $m_-<m$ such that $l_{m_-} = \sigma$ and let $m_-$ be the maximal such index, which directly implies that the only admissible label for $D_{m_-}$ is $\sigma$. 

If there exists some $m_+\geq m$ such that $l_{m_+}= \sigma$, let $m_+$ be the minimal such index. Then, consider the curve $\delta_{m_-,m_+}$ that encircles the $m_-$'th and $m_+$'th framed point and all the framed points in between. Then, using the result of Lemma~\ref{lem:purebraid} 
\be 
E(\hat{\chi}, \hat{\chi}|_{D_{m}\otimes\varepsilon}) = \alpha T_{\delta_{m_-,m_+}}E(\hat{\chi})
\ ,
\ee
where $\alpha$ is a phase. To see how Lemma~\ref{lem:purebraid} is applied, note that the twist of the product of the strands from $m_-$ to $m_+$ will lead to a pure braid thereof multiplied by their respective twist eigenvalues. The twist eigenvalues are absorbed in the phase $\alpha$ and the pure braid is a product of the pure braid generators whose action we described in Lemma~\ref{lem:purebraid}.

If however there does not exist such $m_+\geq m$, then $g\geq 1$ and therefore consider the curve $\delta_{m-1}$ in Figure \ref{fig:mcggenerators-and-ribbongraph}. Then
\be 
E(\hat{\chi},\hat{\chi}|_{D_m\otimes \varepsilon}) = \alpha T_{\delta_{m-1}}E(\hat{\chi}) \ ,
\ee
where the phase $\alpha$ appears according to Lemma \ref{lem:purebraid2}.
\end{proof}

\subsubsection{Property F of Ising Categories}
\label{sec:propf}
In this section we prove that Ising categories have property F with respect to any d-surface. 
\begin{theorem}\label{thm:Ising-propF}
Let $\catname C = \catname C(\zeta, \nu)$ be any Ising-type modular fusion category. Then, $\catname C$ has property F with respect to d-surfaces. 
\end{theorem}

To prove that Ising categories have property F, we will give for any surface $\Sigma_{g,n}$ with simple point labels  $l_1,\dots, l_n$ a certain set $X\equiv X_{g,n}(l_1,\dots, l_n)\subset V_{g,n}$ such that: 

\begin{enumerate}
    \item $X$ is finite 
    \item $X$ spans $V(\Sigma_{g,n})$
    \item $X$ is $\mathrm{PMod}(\Sigma_{g,n})$-invariant
\end{enumerate}
This is sufficient to show property F. The representation image of $\PMod(\Sigma_{g,n})$ is contained in the subgroup of transformations that preserve $X$ by the invariance condition on $X$, \ie $V(\PMod(\Sigma_{g,n})) \subset \{T ~|~ T\in \operatorname{GL}(V_{g,n}), T(X) = X \}$.  
The latter group is a subgroup of the group of bijections on $X$, which follows from the fact that $X$ spans $V_{g,n}$. Since $X$ is a finite set and the group of bijections consists of $|X|!$ elements, the image $V(\PMod(\Sigma_{g,n}))$ has finitely many elements.   

\begin{remark}\label{rem:PropF-non-pure}
\begin{enumerate}
\item As already mentioned in Remark~\ref{rem:propF}, property F
with respect to the pure mapping class group $\PMod$ is equivalent to property F with respect to the non-pure mapping class group $\Mod^\mathcal{C}$.
A similar statement holds with respect to the existence of finite invariant subsets which span the vector space on which the mapping class group acts. 
Namely, consider a surface $\Sigma_{g,n}$ carrying the same point labels $l_1= \dots= l_n$ and a finite $\PMod_{g,n}$-invariant set $X_{g,n}$. Pick a section $s: S_n \rightarrow \Mod_{g,n}$ of the short exact sequence \eqref{eq:ses-non-pure} and define the set 
\be
X_{g,n}^s := \{s(t).x\,|\, x \in X_{g,n},t\in S_n\}\subset V_{g,n}~.
\ee
This set is also finite as both $X_{g,n}$ and $S_n$ are finite. Let $f\in \Mod_{g,n}$ be a mapping class and $t\in S_n$ any permutation. Then, $s(\pi(f)\circ t^{-1})\circ f \circ s(t)\in \ker(\pi) = \PMod_{g,n}$, \ie there exists $g\in \PMod_{g,n}$ such that $f\circ s(t) = s(\pi(f)\circ t)\circ g$. Then, for some $x\in X_{g,n}$ we have 
\be
f\circ s(t).x = s(\pi(f)\circ t)\circ g.x = s(\pi(f)\circ t).x'
\ee 
for some $x'\in X_{g,n}$ which is provided by the $\PMod_{g,n}$-invariance. 
Thus for all $y \in X^s_{g,n}$ also $f.y \in X^s_{g,n}$.

The same can be applied for distinct point labels, when $\Mod_{g,n}$ is the extension of $\PMod_{g,n}$ by some subgroup of the permutation group $S_n$.

\item Similarly, property F with respect to the projective $\Mod(\Sigma)$-action is equivalent to property F with respect to the non-projective action of the extended mapping class group $\widehat{\Mod}(\Sigma)$. This is because the anomaly factor $p_+/p_-$ is a root of unity as pointed out in Remark~\ref{rem:anomalies}. For Ising, that is $p_+/p_- = \theta_{\sigma}^2$.
\end{enumerate}
\end{remark}

Before proving Theorem \ref{thm:Ising-propF} in its generality, we will start with the torus case and the sphere with four framed points to give an idea of how the elements of $X_{g,n}$ will look like. 

\subsubsection*{Torus Case}\label{sec:torus}

The mapping class group of the torus $\PMod_{1,0} = \Mod_{1,0}$ is isomorphic to the group $\mathrm{SL}(2,\mathbb{Z})$ with generators $T$ and $S$. The two generators consist of the Dehn twist $T_\alpha$ along the single meridian $\alpha$, see Figure \ref{fig:torus}, and the $S$-transformation. 

\begin{figure}[t]
\centering
\boxpic{1}{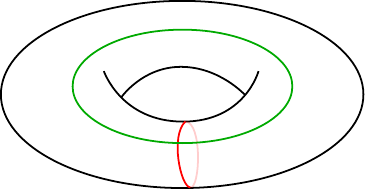}{
\put (50,-4) {$\alpha$}
\put (15,25) {$\beta$}
}
\caption{Dehn twists which generate the mapping class group of the torus}
\label{fig:torus}
\end{figure}

The state space $V_{1,0}$ has basis elements\vspace{1ex} 
\begin{equation}
	e_i ~=~ \boxpic{1}{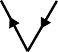}{
	\put (-2,100) {$i$}
\put (88,100) {$i$}
}	\ ,
\end{equation}
where $i\in I=\{\unit,\varepsilon,\sigma\}$ and the action of the generators as described in section \ref{subsec:ising-mcg-action} is given explicitly as: 

\begin{align}
    T_\alpha (e_i) &= \theta_i e_i \nonumber\\
    S(e_\unit) &= \frac{1}{2}(e_\unit + e_\varepsilon + \lambda e_\sigma) \nonumber\\
    S(e_\varepsilon) &= \frac{1}{2}(e_\unit + e_\varepsilon - \lambda e_\sigma) \nonumber\\
    S(e_\sigma) &= \frac{\lambda}{2}(e_\unit - e_\varepsilon)
\end{align}

Now, for the elements $e_\pm:= e_\unit \pm e_\varepsilon$ one can easily check using the above equations that 
\be\label{eq:e_pm-alpha}
T_\alpha e_\pm  = e_\mp,\quad  T_\alpha \lambda e_\sigma = \nu \zeta^{-1} \lambda e_\sigma
\ee
and 
\be
S(e_+) = e_+, \quad S(e_-) = \lambda e_\sigma,\quad S(\lambda e_\sigma) = e_- \ ,
\ee
where the last equality also holds trivially from the fact that $S^2 = \id{}$.
Therefore, the set $X_{1,0}:=\{\theta_\sigma^me_\pm, \theta_\sigma^m \lambda e_\sigma\,| \, m \in \mathbbm{Z}_{16}\}$ is invariant under the generators $T_\alpha$ and $S$ and thus invariant under the projective $\SL$-action.

\subsubsection*{Sphere with four framed points}

We continue by considering the sphere with four framed points with labels $l_1,\dots, l_4$. This is the smallest number of points when the state space $V$ can have $\dim(V)>1$. 

The pure mapping class group of the sphere is generated by Dehn twists $T_{\tau_1}$ and $T_{\tau_2}$ along the curves $\tau_1$ and $\tau_2$ shown in Figure \ref{fig:4sphere} as well as the Dehn twists around each framed point.

\begin{figure}[t]
    \centering
    \boxpic{0.8}{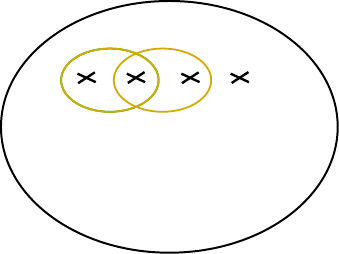}{
    \put (29,35) {$\tau_1$}
    \put (48,35) {$\tau_2$}
}
    \caption{Generators of the (pure) mapping class group of the sphere with four framed points.}
    \label{fig:4sphere}
\end{figure}

The basis elements of $V_{0,4}$ are 
\begin{equation}
	e^{d_3,d_4} ~= ~ \boxpic{1}{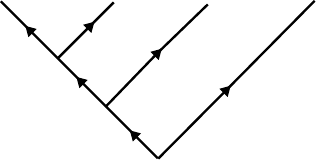}{
		\put (-1,51) {$l_1$}
		\put (65,51) {$l_3$}
		\put (99,51) {$l_4$}
		\put (20,17) {$d_3$}
		\put (36,2) {$d_4$}
		\put (35,51) {$l_2$}
	}
\end{equation}
where $d_3$ and $d_4$ run over admissible labels. In particular, $d_4= l_4$ is the only admissible label and thus we have only dependence on $d_3$. If $l_1 = \dots = l_4 =\sigma$, then $\unit$ and $\varepsilon$ are both admissible labels for $d_3$.  

The action on the basis is given explicitly as
\be 
 T_{\tau_1} e^{d_3} =
 \theta_{d_3} e^{d_3}
\ee
and
\be
T_{\tau_2} e^{d_3}= 
\sum_{d_3', q}{
	\theta_q
	F^{(l_1 l_2 l_3)l_4}_{d_3 q} G^{(l_1 l_2 l_3)l_4}_{q d_3'} e^{d_3'}}~.
\ee
For the action of $T_{\tau_2}$, we notice that if any of the point labels is distinct from $\sigma$, then $e^{d_3}$ is an eigenvector (as the state space is one-dimensional) and the eigenvalue is expressed as an integer power of $\theta_\sigma$. However, if $l_1=\dots= l_4=\sigma$ then $\unit$ and $\varepsilon$ are both admissible labels for $d_3$ and we compute by using twists and F-matrices\footnote{Alternatively, notice that $\theta_\sigma^2 (B^{(\sigma \sigma \sigma)\sigma})^2 = \begin{pmatrix}
		0 & 1\\
		1& 0
\end{pmatrix}$ using \eqref{eq:Ising-B-matrix-2}.}
\be
T_{\tau_2}e^{\varepsilon_a} = \frac{1}{2}(1+ (-1)^{a+1})e^{\unit} + \frac{1}{2}(1+ (-1)^a)e^\varepsilon = e^{\varepsilon_{a+1}}~.
\ee

The above calculations imply that 
\be\label{eq:X-4-sphere} 
X_{0,4} = \{\theta_\sigma^k e^{d_3}\,|\, k\in \mathbbm{Z}_{16}, d_3\}
\ee
is $\PMod_{0,4}$-invariant. Notice however that it is not $\Mod_{0,4}$-invariant (see Remark~\ref{rem:PropF-non-pure}). 

\subsubsection*{Closed Surfaces}

We choose to prove property F with respect to closed surfaces as the proof can be easily extended to the case with framed points. Recall the basis elements obtained from colourings of the graph $\Gamma_{g,0}$ in \eqref{eq:basisgraph}. For any colouring $\chi \in \colGo$ we consider an alternate colouring $\Tilde{\chi}$, in which we change every $J_m$ labelled by $\unit$ to the plus label $+$ and every $J_m$ labelled by $\varepsilon$ to the minus label $-$. This gives a new colouring set $\tcolGo$ consisting of these alternate colourings. Let $\Tilde{\chi}$ be a colouring in $\tcolGo$ with $J_{m_1},\dots,J_{m_L}$ the maximal subsequence of edges labelled by signs $\pm$ with the rest of $I$-edges labelled by $\sigma$. We define the associated vector in $V_{g,0}$ by  
\be\label{eq:Tildechi} 
[\Tilde{\chi}] = \sum_{p_1,\dots,p_L= 0,1}{\Tilde{\chi}(J_{m_1})^{p_1}\cdots \Tilde{\chi}(J_{m_L})^{p_L}\hat{\chi}|_{J_{m_1}=\varepsilon_{p_1},\dots, J_{m_L} = \varepsilon_{p_L}}} \ ,
\ee
where $ \hat{\chi}|_{J_{m_1}=\varepsilon_{p_1},\dots, J_{m_L} = \varepsilon_{p_L}}$ is the basis element corresponding to the colouring in $\colGo$ obtained from $\Tilde{\chi}$ by changing the indicated $\pm$-labelled $I$ edges accordingly. This notation is the higher genus analogue of $e_\pm$ for the torus case. 

\begin{proposition}\label{prop:X_g0}
The finite set \[X_{g,0} = \{\theta_\sigma^k \lambda^{m_{\Tilde{\chi}}} [\Tilde{\chi}]\,| \, k\in \mathbbm{Z}_{16}, \Tilde{\chi}\in \tcolGo\} \ ,\]where $m_{\Tilde{\chi}}= |\{k\,|\, \Tilde{\chi}(J_k) = \sigma\}|$, spans the state space and is $\Mod_{g,0}-$invariant. 
\end{proposition}
\begin{proof}
\begin{enumerate}
    \item The fact that this set spans the state space comes from the definition of the $\pm$-notation. Let $\hat\chi$ be some basis element of the state space and let $J_{m_1},\dots, J_{m_L}$ be the maximal subsequence of $I$-edges such that $\chi(J_{m_1}) = \varepsilon_{p_1},\dots, \chi(J_{m_L}) = \varepsilon_{p_L}$ (all non-sigma labels). Then, by inverting \eqref{eq:Tildechi}: 
\[ \hat\chi = \frac{1}{2^L}\sum_{\nu_1,\dots,\nu_L = \pm}{\nu_1^{p_1}\cdots \nu_L^{p_L} [\Tilde{\chi}|_{J_{m_1}= \nu_1,\dots, J_{m_L}=\nu_L}]} \ ,\]
where $\Tilde{\chi}|_{J_{m_1}= \nu_1,\dots, J_{m_L}=\nu_L}$ is the alternate colouring obtained from $\chi$ by changing non $\sigma$-labels of $I$-edges into $\pm$-labels.  

\item We now prove $\Mod_{g,0}$-invariance. 
It is invariant under $T_{\alpha_m}$'s as on the basis, one has from \eqref{eq:alpha-action}: 
\[T_{\alpha_m} \hat{\chi} = \theta_{\chi(J_m)}\hat{\chi}\] and from the observation in \eqref{eq:e_pm-alpha} it follows for an element of $X_{g,0}$:
\[ T_{\alpha_m} \theta_\sigma^k \lambda^{m_{\Tilde{\chi}}} [\Tilde{\chi}] = \left\{
\begin{array}{cc}
    \theta_\sigma^{k+1} \lambda^{m_{\Tilde{\chi}}} [\Tilde{\chi}] &\text{if}\quad\Tilde{\chi}(J_m) = \sigma \\
    \theta_\sigma^k \lambda^{m_{\Tilde{\chi}}} [\Tilde{\chi}|_{J_m = -\Tilde{\chi}(J_m)}] &\text{if}\quad\Tilde{\chi}(J_m) = \pm\\
\end{array}\right.\]
Recall from \eqref{eq:Saction-1}, \eqref{eq:Saction-sigma} and \eqref{eq:Saction-sigma2} how $S_m$ acts on our fixed basis. Then, one can easily check that the action on the set elements is given by 
\[ S_m \theta_\sigma^k \lambda^{m_{\Tilde{\chi}}} [\Tilde{\chi}] = \left\{\begin{array}{cc}
    \theta_\sigma^k \lambda^{m_{\Tilde{\chi}}} [\Tilde{\chi}] & \text{if}\quad\Tilde{\chi}(J_m) = + \\
    \theta_\sigma^k \lambda^{m_{\Tilde{\chi}} +1} [\Tilde{\chi}|_{J_m = \sigma}] & \text{if}\quad\Tilde{\chi}(J_m) = -\\
    \theta_\sigma^k \lambda^{m_{\Tilde{\chi}} -1} [\Tilde{\chi}|_{J_m = -}] & \text{if}\quad\Tilde{\chi}(J_m) = \sigma, \Tilde{\chi}(K_m) = \unit\\
    \zeta^2\theta_\sigma^k \lambda^{m_{\Tilde{\chi}}} [\Tilde{\chi}|_{J_m =\sigma}] &\text{if}\quad \Tilde{\chi}(J_m) = \sigma, \Tilde{\chi}(K_m) =\varepsilon
\end{array}\right. \]
which means that the set $X_{g,0}$ is $S_m$-invariant.

Recall from \eqref{eq:Tgamma-nonsigma} and \eqref{eq:Tgamma-sigma} how $T_{\gamma_m}$ acts on the fixed basis. This results in \[ T_{\gamma_m} \theta_\sigma^k \lambda^{m_{\Tilde{\chi}}} [\Tilde{\chi}] = \left\{\begin{array}{cc}
    \theta_\sigma^k \lambda^{m_{\Tilde{\chi}}} [\Tilde{\chi}|_{J_m = -\Tilde{\chi}(J_m), J_{m+1} = -\Tilde{\chi}(J_{m+1})}] &\text{if}\quad \Tilde{\chi}(J_m),\Tilde{\chi}(J_{m+1}) = \pm \\
    \theta_\sigma^{k+1} \lambda^{m_{\Tilde{\chi}}} [\Tilde{\chi}]
    & \text{if}\quad\Tilde{\chi}(J_m) = \pm, \Tilde{\chi}(J_{m+1}) = \sigma\\
    \theta_\sigma^k \lambda^{m_{\Tilde{\chi}}} [\Tilde{\chi}|_{\{K_m,K_{m+1},A_{k+1}\}\otimes \varepsilon}] & \text{if}\quad\Tilde{\chi}(J_m)=\Tilde{\chi}(J_{m+1}) = \sigma\\
\end{array}\right. \]
This concludes $T_{\gamma_m}-$invariance.

Invariance under the generators implies invariance under $\Mod_{g,0}$, even though the representation is only projective. This is because the projective factors can be expressed as integer powers of $\theta_\sigma$ as discussed in Remark~\ref{rem:PropF-non-pure}.
\end{enumerate}
\end{proof}

\subsubsection*{General Case}

The next statement completes the proof of Proposition~\ref{thm:Ising-propF}.

\begin{proposition}
The finite set \[X_{g,n} = \{\theta_\sigma^k \lambda^{m_{\Tilde{\chi}}} [\Tilde{\chi}]\,| \, k\in \mathbbm{Z}_{16}, \Tilde{\chi}\in \tcolGo\} \ ,\]where $m_{\Tilde{\chi}}= |\{k\,|\, \Tilde{\chi}(J_k) = \sigma\}|$, spans the state space and is $\PMod_{g,n}-$invariant. 
\end{proposition}
\begin{proof}
The invariance under the generators $T_{\alpha_k}, S_k$ and $T_{\gamma_k}$ follows in exactly the same way as in Proposition~\ref{prop:X_g0} as the action of these generators does not depend on or affect the labels of $L_1,\dots,L_n,D_3,\dots,D_{n+1}$. It only remains to find out what happens when we act with the $T_\delta$ generators. 

To show invariance under $T_{\delta_k}$'s, it is sufficient to prove that the set is invariant under pure braids in $L_{1},\dots, L_{n}, J_1$ strands. The proof follows from Lemmata~\ref{lem:purebraid} and \ref{lem:purebraid2}.
\end{proof}

\section{Applications to 3D gravity}\label{sec:physics-review}

In this section we discuss two applications of the irreducibility and finiteness properties of mapping class group representation to 3D\,gravity. The first application concerns attempts to give meaning to the gravity partition function, and the second application concerns invertible global symmetries.

\subsection{Path integral formulation of quantum gravity}

For a fixed conformal boundary, \ie a Riemann surface $\Sigma$, 
the partition function of pure Einstein gravity takes the form
\be\label{eq:path-int}
Z_\mathrm{grav}(\Sigma) = \sum_{M\text{ topologies}}\int{\mathcal{D}g\, e^{iS[g]}} \ ,
\ee
where $S[g]$ is the Einstein-Hilbert action, the sum is over diffeomorphism classes of smooth $3$-manifolds $M$ with conformal boundary $\Sigma$ and we perform a path integral over Riemannian metrics $g$ with conformal boundary $\Sigma$. 
There are various approaches with which one can try and make sense of \eqref{eq:path-int}. 

\medskip

In one and two dimensions, the sum over topologies can be performed explicitly. This case has been studied in detail in \cite{BMoore22}.

\medskip

In three dimensions, as in \cite{Maloney:2007ud,Castro:2011zq} our starting assumption will be that the path integral in \eqref{eq:path-int} produces an amplitude for a 2D\,CFT assigned to the fixed conformal boundary $\Sigma$. This is motivated by the $\AdS/\operatorname{CFT}_2$ correspondence and by the observation that the Poisson algebra of certain diffeomorphisms of $\AdS$ gives two copies of the Virasoro algebra of central charge $c = \frac{3l}{2G}$ \cite{BH}, where $l$ is the AdS-radius and $G$ is Newton's constant.

We formalise this as follows: Write $W(\Sigma)$ for the vector space of bilinear combinations of holomorphic and anti-holomorphic conformal blocks
of the 2D\,CFT. In the RCFT examples we will turn to shortly, this is finite-dimensional, but more generally it could be infinite. Let $Z(M) \in W(\Sigma)$ be the CFT amplitude which replaces the path integral over metrics for a fixed $M$ in \eqref{eq:path-int}. 
The sum over $M$ in \eqref{eq:path-int} is by definition the same as the sum over isomorphism classes of bordisms $M : \emptyset \to \Sigma$, i.e.\ a sum over $\mathrm{Bord}(\emptyset,\Sigma)$, the hom-space in the bordism category as in Section~\ref{sec:RT}.
Thus \eqref{eq:path-int} becomes
\be\label{eq:bordism-sum}
Z_\mathrm{grav}(\Sigma) = \sum_{M \in \mathrm{Bord}(\emptyset,\Sigma)} Z(M)
 \ .
\ee
As it stands, this sum is still ill-defined.

\begin{figure}[bt]
    \centering
    \begin{equation*}
    	M_\gamma =~ \boxpic{0.8}{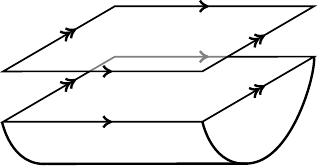}{
    		\put (95,38) {$\uparrow$}
    		\put (103,40) {$\gamma$} 
    	}
    \end{equation*}
    \caption{Gluing the solid torus to the boundary torus via a mapping class group element element $\gamma \in \mathrm{SL}(2,\mathbb Z)$.}
    \label{fig:M-gamma}
\end{figure}    

The semi-classical approach (corresponding to large central charge of the boundary CFT) of \cite{Maloney:2007ud} suggests that the dominating contributions to \eqref{eq:path-int} are the classical solutions. When $\Sigma$ is a torus, these can be obtained by elements of $\SL$ which specify how thermal $\AdS$ (a solid torus) is glued to the conformal boundary (Figure~\ref{fig:M-gamma}). This results in a sum over the mapping class group of some initial value (the seed), for example the value assigned to thermal $\AdS$. Such mapping class group averages, also for more general geometries, have been investigated in \cite{Castro:2011zq,Jian:2019ubz,Meruliya:2021utr,Raeymaekers:2021ANO}.
As already mentioned in Remark~\ref{rem:mcg-avg-def}, if the boundary CFT is rational, for the torus the sum over $\SL$ can be reduced to a finite sum. One finds, depending on the RCFT one starts from, that the result is the partition function of a 2D\,CFT, or an ensemble average, or a bilinear combination of characters that cannot be interpreted as an ensemble. See also \cite{CollierM} for ensemble averages over a continuous moduli space of CFTs.

\medskip

We would like to apply irreducibility of mapping class group actions and property F to the problem of making sense of \eqref{eq:bordism-sum}. 
To this end, note that $\mathrm{Bord}(\emptyset,\Sigma)$ carries a natural action of $\Mod(\Sigma)$ by modifying the boundary parametrisation. We can thus decompose $\mathrm{Bord}(\emptyset,\Sigma)$ into $\Mod(\Sigma)$-orbits $\mathcal{B}_\alpha$, $\alpha \in \mathcal{O}$, where $\mathcal{O} = \mathrm{Bord}(\emptyset,\Sigma) / \Mod(\Sigma)$ is the set of orbits.
Let $M_\alpha \in \mathcal{B}_\alpha$ be a choice of representative for each orbit. Then we can rearrange \eqref{eq:bordism-sum} as
\be\label{eq:MCG-orbit-bordism-sum}
Z_\mathrm{grav}(\Sigma) 
= \sum_{\gamma \in \Mod(\Sigma)} \sum_{\alpha \in \mathcal{O}} Z(\gamma.M_\alpha)
= \sum_{\gamma \in \Mod(\Sigma)} \gamma.\Big( \sum_{\alpha \in \mathcal{O}} Z(M_\alpha) \Big)
\ .
\ee
Since we just rearranged \eqref{eq:bordism-sum}, the above sum remains ill-defined. But this form of rewriting \eqref{eq:bordism-sum} illustrates two points.

Firstly, any reasonable way to make sense of \eqref{eq:bordism-sum} should result in a vector $Z_\mathrm{grav}(\Sigma) \in W(\Sigma)$ that is invariant under the mapping class group action. Let us assume that the boundary CFT is rational, i.e.\ that it contains a holomorphic and an antiholomorphic copy of some rational VOA $\mathcal{A}$ in its space of fields. 
Set $\mathcal{C} = \operatorname{Rep}(\mathcal{A})$. If the projective mapping class group representations obtained from $\mathcal{C}$ are irreducible, by Lemma~\ref{lem:mcg-invariant-subsp} the subspace of invariants in $W(\Sigma) = V^\mathcal{C}(\widehat\Sigma)$ is one-dimensional. Hence $Z_\mathrm{grav}(\Sigma)$ must be proportional to the CFT correlator $\operatorname{Cor}_\mathcal{A}(\Sigma)$ of the unique 2D\,CFT with holomorphic and anti-holomorphic symmetry $\mathcal{A}$,
\be
	Z_\mathrm{grav}(\Sigma) \propto \operatorname{Cor}_\mathcal{A}(\Sigma) ~.
\ee
The uniqueness of the 2D\,CFT under the above irreducibility assumption is the main result of \cite{RR}. 
We summarise the first point illustrated by \eqref{eq:MCG-orbit-bordism-sum} as:
\begin{center}
\framebox[1.05\width]{\parbox{.9\textwidth}{
\ \\[0em]
2D\,CFTs which have irreducible chiral mapping class group representations are good candidates to investigate properties of 3D\,gravity without the need to explicitly regularise \eqref{eq:bordism-sum}.
\\[-.6em]
}}
\end{center}

Secondly, if $\mathcal{C}$ has property F, the sum over $\Mod(\Sigma)$ in \eqref{eq:MCG-orbit-bordism-sum} factors through a finite sum over the image of $\Mod(\Sigma)$ in $\mathrm{GL}(W(\Sigma))$. 
Thus the extra information we have under property F is an explicit projector
\be\label{eq:invariance-projector}
	\langle-\rangle_\Sigma \colon
	W(\Sigma) \longrightarrow W(\Sigma)^{\Mod(\Sigma)}
	\quad , \quad
	x \mapsto \langle x \rangle_\Sigma ~,
\ee
where again $W(\Sigma) =  V^\mathcal{C}(\widehat\Sigma)$ and the average $\langle x \rangle_\Sigma$ was introduced in Definition~\ref{def:mcgaverage}.

The nicest situation arises if $\mathcal{C}$ has property F and irreducible mapping class group representations at the same time, like e.g.\ Ising-type CFTs (Theorems~\ref{thm:Ising-irred} and~\ref{thm:Ising-propF}). In this case one can apply the projector \eqref{eq:invariance-projector} to any state $x \in W(\Sigma)$ and obtain either zero or
a state proportional to the (unique in this case) 2D\,CFT correlator on $\Sigma$ (Theorem~\ref{thm:mcgavg-cor}):
\be
\text{for all } x \in W(\Sigma) ~: \quad 
\langle x \rangle_\Sigma 
= \lambda \operatorname{Cor}_\mathcal{A}(\Sigma) 
\quad \text{for some}~~\lambda \in \mathbb{C} \ .
\ee

This expression has an appealing counterpart in 3D\,TQFT. Namely, when the VOA $\mathcal{A}$ is rational and hence $\mathcal{C}$ is an MFC, the relevant 3D\,TQFT is the RT TQFT $\mathcal{Z}^\mathcal{C\boxtimes C^\mathrm{rev}}$, see Section~\ref{subsec:RCFT}. Any state $x \in W(\Sigma) = V^\mathcal{C}(\widehat\Sigma)$ can be written as a handlebody with an appropriate coupon as in Figure~\ref{fig:handlebodycoupon} (but for $\mathcal{C}\boxtimes \mathcal{C}^\mathrm{rev}$, not just for $\mathcal{C}$). Recall from~\eqref{eq:RCFT-cor-via-bnd} and Figure~\ref{fig:cor-bnd-def} the relation between 2D\,CFT correlators and three-manifolds with boundary. Putting all this together, we get, for example,
\begin{align}
&\Bigg\langle \mathcal{Z}^\mathcal{C\boxtimes C^\mathrm{rev}}\Bigg(
\boxpic{0.65}{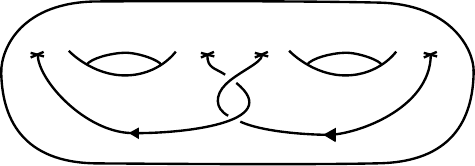}{
\put (16.5,12) {\footnotesize$i\times i^\ast$}
\put (68,12) {\footnotesize$j\times j^\ast$}
}
 \Bigg) \Bigg\rangle_\Sigma 
 \nonumber
\\[.5em]
&\quad \propto \quad
\mathcal{Z}^\mathcal{C\boxtimes C^\mathrm{rev}}
\hspace{-.2em}
\left( 
\boxpic{0.65}{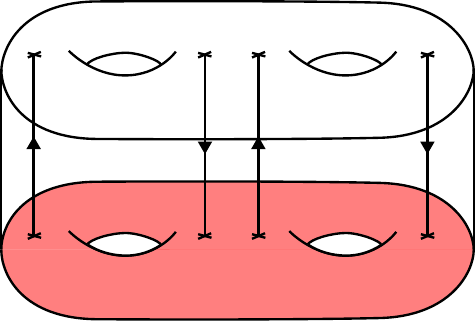}{
\put (10,31.5) {\footnotesize$i\times i^\ast$}
\put (56,31.5) {\footnotesize$j\times j^\ast$}
}
\right) ~.
\end{align} 
The geometry on the lhs (the seed) is a handlebody with boundary $\Sigma$ and with embedded $\mathcal{C}\boxtimes \mathcal{C}^\mathrm{rev}$-labelled ribbon graph (in this example without coupons). The geometry on the rhs is $\Sigma \times [0,1]$.
We summarise this as:
\begin{center}
\framebox[1.05\width]{\parbox{.9\textwidth}{
\ \\[0em]
Suppose the boundary 2d\,CFT is rational with corresponding 3D\,TQFT $\mathcal{Z}^\mathcal{C\boxtimes C^\mathrm{rev}}$, such that $\catname C$ satisfies irreducibilty and property F. 
Then the mapping class group average removes all topological information present in the seed and introduces a boundary to the 3D\,TQFT to which all Wilson lines connect orthogonally. 
\\[-.6em]
}}
\end{center}
Under the present assumptions, the boundary condition to the RT TQFT for $\mathcal{C}\boxtimes \mathcal{C}^\mathrm{rev}$ is unique (since the unique indecomposable semisimple $\mathcal{C}$-module category is $\mathcal{C}$ itself \cite{RR}). And for the concrete example of a ribbon graph on the lhs above (no coupons and no loops), by Lemma~\ref{lem:tangle-pairing} we know that the average will be non-zero.

\subsection{Absence of invertible global symmetries}\label{subsec:global-sym}

In this section,
we study invertible global symmetries\footnote{Topological field theories also allow for non-invertible topological defects, which in this context are often referred to as symmetries, while here we are only concerned with invertible topological defects.
}
in the topological bulk theory of $\mathcal{Z}(\mathcal{C})$.
It is a conjecture by \cite{Harlow:2018jwu} that quantum gravitational theories do not admit any invertible global symmetries.
In this section we investigate this question for RT TQFTs for $\mathcal{Z}(\catname C)$ in the case that the mapping class group representations of $\catname C$ are irreducible.

\medskip

Invertible global symmetries of TQFTs correspond to invertible module categories. 
For an MFC $\catname C$, the Drinfeld centre $\mathcal{Z}(\catname C)$ obeys the factorisation property 
\be\label{eq:ZC-is-CC}
\mathcal{ Z}(\catname C) \simeq \catname C \boxtimes \catname C^\mathrm{rev} \ .
\ee
Therefore, $\mathcal{ Z}(\catname C)$-module categories are the same as $\catname C$-$\catname C$-bimodule categories.
Equivalence classes of invertible module categories over $\mathcal{Z}(\catname C)$ by definition form the Picard group $\operatorname{Pic}(\mathcal{Z}(\mathcal{C}))$ and similarly invertible bimodule categories over $\catname C$ define the Brauer-Picard group $\operatorname{BrPic}(\mathcal{C})$. The equivalence \eqref{eq:ZC-is-CC} results in a group isomorphism (see \cite{ENOM})
\begin{equation}\label{eq:BrPic-Pic}
    \operatorname{Pic}(\mathcal{Z}(\mathcal{C}))\cong \operatorname{BrPic}(\mathcal{C}) \ .
\end{equation}

We make the following assumption:
\begin{equation}
    \text{$\catname C$ has no indecomposable semisimple module categories other than itself.}
    \tag{$*$}
\end{equation}
By \cite[Thm.~5.1(v2)]{RR},
irreducibility of the mapping class group representations of $\catname C$ implies $(*)$ (see Remark~\ref{rem:thm-mcg-avg-cor}), but not conversely. 

\medskip

Condition $(*)$ implies that the TQFT $\mathcal{Z}^{\catname C}$ has no non-trivial surface defects, and in particular no non-trivial invertible global symmetries. However, as we saw in the previous section, the TQFT relevant to obtain 2D\,CFT correlators on the boundary is $\mathcal{Z}^\mathcal{C\boxtimes C^\mathrm{rev}}$ and not $\mathcal{Z}^{\catname C}$. Unless $\catname C \cong \Vect$, $\mathcal{C\boxtimes C^\mathrm{rev}}$ \textit{always} has at least one non-trivial module category, namely $\catname C$, and moreover the corresponding surface defect of $\mathcal{Z}^\mathcal{C\boxtimes C^\mathrm{rev}}$ is always non-invertible.
Thus for $\catname C \ncong \Vect$, $\mathcal{Z}^\mathcal{C\boxtimes C^\mathrm{rev}}$ always has at least one non-invertible global symmetry.
However, it still makes sense to ask if under additional assumptions like $(*)$ or irreducibility, $\mathcal{Z}^\mathcal{C\boxtimes C^\mathrm{rev}}$ has no non-trivial \textit{invertible} global symmetries. By \eqref{eq:BrPic-Pic} this amounts to the question whether $\operatorname{BrPic}(\mathcal{C})$ is trivial or not. To investigate this, we start with the following observation:

\begin{proposition}\label{prop:no-invertible}
	Let $\catname C$ be a braided fusion category satisfying $(*)$.
 Then there is an isomorphism
 \[\operatorname{BrPic(\mathcal{C})} \cong \Aut_\otimes (\mathcal{C})\]
 between the Brauer-Picard group and the group of tensor auto-equivalences.
\end{proposition}

The proof of Proposition~\ref{prop:no-invertible} will be given after the next two lemmas.

Recall from \cite[Sec.~4.3]{ENOM} that an invertible $\mathcal{C}\mbox{-}\mathcal{C}$-bimodule category $\mathcal{M}$ is \textit{quasi-trivial} if $_{\mathcal{C}}\mathcal{M} \simeq{} _{\mathcal{C}}\mathcal{C}$ as left $\mathcal{C}$-module categories. 
Equivalence classes of quasi-trivial bimodule categories form a subgroup $\operatorname{QTriv}(\mathcal{C})$ in $\operatorname{BrPic}(\mathcal{C})$ and we have the following lemma due to \cite[Sec.\ 4.3]{ENOM} (see also \cite[Lem.\ 7.25]{Romaidis-Thesis} 
for details on the proof):
\begin{lemma}\label{lem:quasi-tr-outer}
Let $\catname C$ be a fusion category. There is an isomorphism between the group of outer equivalences and $\operatorname{QTriv}(\mathcal{C})$:
\begin{equation}
    \operatorname{Out}_\otimes(\mathcal{C}) \xrightarrow{~\sim~} \operatorname{QTriv}(\mathcal{C}) \quad , \quad [\phi] \longmapsto \left[{}_{\mathcal{C}}\mathcal{C}_{\mathcal{\phi}}\right] \ ,
\end{equation}
where ${}_{\mathcal{C}}\mathcal{C}_{\mathcal{\phi}}$ has an underlying left regular $\mathcal{C}$-module structure and the right action is twisted by the tensor auto-equivalence $\phi \in \Aut_{\otimes}(\mathcal{C})$, \ie $M\triangleleft X := M\otimes \phi(X)$.
\end{lemma}
In particular, every quasi-trivial bimodule category is up to equivalence uniquely determined by an outer equivalence. 

\begin{lemma}\label{lem:BrPic=Out}
Let $\catname C$ be a fusion category satisfying $(*)$.
Then, every invertible bimodule category is quasi-trivial, \ie $\operatorname{QTriv}(\mathcal{C}) = \operatorname{BrPic}(\mathcal{C})$
\end{lemma}
\begin{proof}
By \cite[Cor.~4.4]{ENOM} any invertible bimodule category $\mathcal M$ is indecomposable as a left module category. This follows directly from the invertibility property 
$\mathcal{M} \boxtimes_{\catname C}\mathcal{M}^\text{op} \simeq \catname C$
as bimodule categories. However, since $\catname C$ has only the left regular category $\catname C$ as an indecomposable semisimple
left module category up to equivalence, $\mathcal{ M}$ is equivalent to $\catname C$ as a left module category and thus quasi-trivial. 
\end{proof}

\begin{proof}[Proof of Proposition~\ref{prop:no-invertible}]
Since $\catname C$ is braided, inner equivalences are trivial and hence $\mathrm{Out}(\catname C) \cong \Aut_\otimes(\catname C)$. The result follows from Lemmas~\ref{lem:quasi-tr-outer} and~\ref{lem:BrPic=Out}.
\end{proof}

Altogether we see that under assumption $(*)$,  invertible global symmetries of $\mathcal{Z}^\mathcal{C\boxtimes C^\mathrm{rev}}$ are precisely given by tensor auto-equivalences $\Aut_\otimes (\mathcal{C})$ of $\mathcal{C}$. 

\begin{remark}\label{rem:no-invertible}
\begin{enumerate}
    \item 
    If $\mathcal{C}$ is an MFC and satisfies $(*)$, then 
    the group $\Aut_\mathrm{br}(\mathcal{C})$ of \textit{braided} tensor auto-equivalences of $\mathcal{C}$ is trivial. Indeed, by \cite[Thm.~5.2]{ENOM} the Picard group $\operatorname{Pic}(\mathcal{C})$ is isomorphic to $\Aut_\mathrm{br}(\mathcal{C})$ but by hypothesis $\operatorname{Pic}(\mathcal{C})$ is trivial. 
    To show that also $\Aut_\otimes(\mathcal{C})=1$ one needs to show that every tensor auto-equivalence already preserves the braiding. We do not know whether or not this is a consequence of $(*)$ or of irreducibility, but it holds in examples, as we summarise next.
    
    \item    
    In \cite{Edie} the groups of tensor auto-equivalences $\Aut_{\otimes}(\catname C)$ and braided auto-equivalences $\Aut_\mathrm{br}(\mathcal{C})$ are studied for the MFCs $\catname C = \catname C(\mathfrak{g},k)$ associated to certain
    affine Lie algebras $\widehat{\mathfrak{g}}$ at levels $k \in \mathbb{Z}_{>0}$.
    Typically, $\Aut_\mathrm{br}(\mathcal{C})\subsetneq\Aut_{\otimes}(\catname C)$, but \eg for 
    $\mathfrak{sl}_2$ and $k=2$ or $k$ odd both groups are trivial. 
    In fact, 
    \be 
    \operatorname{BrPic}(\mathcal{C}(\mathfrak{sl}_2,k)) = 1  \quad
    \text{if and only if}~~
    k=2\text{ or $k$ odd.}
    \ee  
    These are precisely the levels $k$ for which $\mathcal{C}(\mathfrak{sl}_2,k)$ satisfies $(*)$, see \cite{Ostrik:2001}.
    \item 
    Our starting point to quantum gravity was via sums
    over topologies which led us to the study of mapping class group averages and finally to the study of invertible global symmetries in this section. 
    In \cite{Benini2023}, conversely, the demand to have no global symmetries is taken as the starting point. After gauging a one-form global symmetry, the resulting theory has no global symmetries and therefore serves as a candidate gravity theory.

    The absence of also non-invertible symmetries has recently been studied and linked to the completeness hypothesis for the spectrum \cite{RudeliusShao,heidenreich} and to the weak gravity conjecture \cite{CordovaKantaroRudelius}. 
\end{enumerate}
\end{remark}

In all examples we are aware of for which $\catname C$ has irreducible mapping class group representations (namely Ising-type categories and $\mathcal{C}(\mathfrak{sl}_2,k)$ for $k+2$ prime or (conjecturally) the square of a prime, see Section~\ref{sec:irred-points-or-not}), part 2 of the above remark gives $\Aut_\otimes (\mathcal{C})=1$, and so in these cases
$\mathcal{Z}^\mathcal{C\boxtimes C^\mathrm{rev}}$ indeed has no invertible global symmetries.

\newpage

\appendix

\section{R,F,B-matrices}
\label{app:rfb-matrices}

Let $\catname C$ be a fusion category
and $i,j,k$ simple representatives in $I$. 
Recall that fusion coefficients are defined as  
\be\label{eq:fusion-coeff}
N_{ij}^k := \dim \catname C(k,i\otimes j)
\ee 
and satisfy 
\be 
N_{ij}^k = N_{ji}^k = N_{i\overline{k}}^{\overline{j}} = N_{\overline{i}\overline{j}}^{\overline{k}}~.
\ee 
The fusion basis elements are graphically represented by
\be\label{eq:fusion-basis}
\boxpic{0.8}{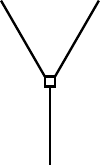}{
\put (27,-19) {$k$}
\put (-3,102) {$i$}
\put (55,105) {$j$}
\put (10,45) {$\alpha$}
},\quad \boxpic{0.8}{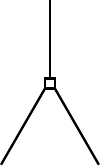}{\put (27,102) {$k$}
\put (-3,-18) {$i$}
\put (55,-18) {$j$}
\put (10,45) {$\bar\alpha$}} 
\ee 
where $\alpha \in \{1,\dots, N_{ij}^k\}$ and are dual in the sense that $\bar\alpha \circ \beta = \delta_{\alpha, \beta}\id{k}$.  
When considering the fusion of four labels $i,j,k,l$ in $I$, we have two natural decompositions, namely
\begin{align}
    \catname C (l, i \otimes j \otimes k) &\cong \bigoplus_{m\in I}\catname C(l, m\otimes k)\otimes_\mathbbm{k} \catname C(m,i\otimes j)\\
    & \cong \bigoplus_{n\in I} \catname C(l,i\otimes n)\otimes_\mathbbm{k} \catname C(n,j\otimes k)
\end{align}
which gives two bases respectively. The transition matrix between these two bases is called an $F$-matrix and is defined by the relation (cf.\ \cite[Eq.\,(2.39)]{FRS1})
\be\label{eq:F-matrix}
(\alpha \otimes \id{k})\circ \beta = \sum_{n,\gamma,\delta}
F^{(ijk)l}_{\gamma n \delta,\,\alpha m \beta} \, (\id{i}\otimes \delta)\circ \gamma
~.
\ee
The matrix elements of the transformation inverse to $F$ will be denoted by $G$:
\be\label{eq:G-matrix}
(\id{i}\otimes \delta)\circ \gamma
= \sum_{m,\alpha,\beta}
G^{(ijk)l}_{\alpha m \beta,\,\gamma n \delta} \, 
(\alpha \otimes \id{k})\circ \beta 
~.
\ee
The string diagrams for equations \eqref{eq:F-matrix} and \eqref{eq:G-matrix} are 
\vspace{1ex}
\begin{align}
&\boxpic{0.8}{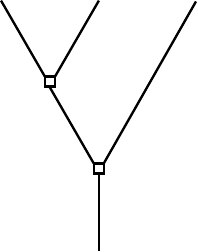}{
\put (74,102) {$k$}
\put (-3,102) {$i$}
\put (35,102) {$j$}
\put (8,60) {$\alpha$}
\put (28,26) {$\beta$}
\put (37,-12) {$l$}
\put (18,42) {$m$}
} = \sum_{n,\gamma,\delta}
F^{(ijk)l}_{\gamma n \delta,\,\alpha m \beta} \, \boxpic{0.8}{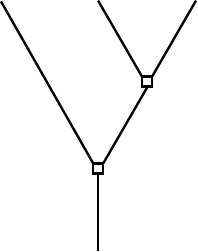}{
\put (74,102) {$k$}
\put (-3,102) {$i$}
\put (35,102) {$j$}
\put (62,60) {$\delta$}
\put (42,26) {$\gamma$}
\put (37,-12) {$l$}
\put (50,42) {$n$}
} ~,
\\[3ex]
&\boxpic{0.8}{figures/Fmatrix2.pdf}{
	\put (74,102) {$k$}
	\put (-3,102) {$i$}
	\put (35,102) {$j$}
	\put (62,60) {$\delta$}
	\put (42,26) {$\gamma$}
	\put (37,-12) {$l$}
	\put (50,42) {$n$}
}
= \sum_{m,\alpha,\beta}
G^{(ijk)l}_{\alpha m \beta,\,\gamma n \delta} \, 
\boxpic{0.8}{figures/Fmatrix1.pdf}{
	\put (74,102) {$k$}
	\put (-3,102) {$i$}
	\put (35,102) {$j$}
	\put (8,60) {$\alpha$}
	\put (28,26) {$\beta$}
	\put (37,-12) {$l$}
	\put (18,42) {$m$}
}~.
\vspace{1ex}
\end{align}
In a spherical fusion category we can compute them by:
\be 
F^{(ijk)l}_{\gamma n \delta,\,\alpha m \beta} = \frac{1}{d_l} \boxpic{1}{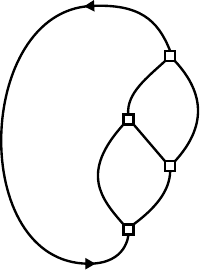}{
	\put (75,58) {$k$}
	\put (31,27) {$i$}
	\put (50,41) {$j$}
	\put (67,33) {$\delta$}
	\put (51,10) {$\gamma$}
	\put (59,20) {$n$}
	\put (38,56) {$\bar\alpha$}
	\put (52,77) {$\bar\beta$}
	\put (62,90) {$l$}
	\put (53,63) {$m$}
} ~, \quad 
G^{(ijk)l}_{\alpha m \beta,\,\gamma n \delta} =
\frac{1}{d_l} \boxpic{1}{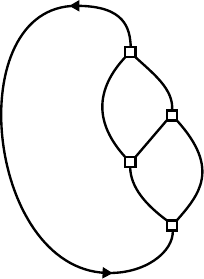}{
	\put (74,30) {$k$}
	\put (31,57) {$i$}
	\put (51,52) {$j$}
	\put (37,38) {$\alpha$}
	\put (53,12) {$\beta$}
	\put (53,28) {$m$}
	\put (66,56) {$\bar\delta$}
	\put (50,80) {$\bar\gamma$}
	\put (44,95) {$l$}
	\put (60,72) {$n$}
} 
\ee
Let $\catname{C}$ be a braided fusion category. The $R$-matrix describes how the fusion basis changes under the braiding of $\catname C$. Namely, it is defined as
\be\label{eq:R-matrix}
c_{i,j}\circ \alpha = 
    \sum_{\beta}
    R^{(ij)k}_{\beta\alpha} 
\, \beta~
\quad , \quad
c_{j,i}^{-1}\circ \alpha = 
    \sum_{\beta}
    R^{-\,(ij)k}_{\beta\alpha} 
\, \beta~
\ee 
and graphically: 
\vspace{1ex}
\be 
\boxpic{0.8}{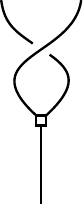}{
	\put (17,-13) {$k$}
	\put (-3,103) {$j$}
	\put (40,103) {$i$}
	\put (5,37) {$\alpha$}
} = 
\sum_{\beta}
R^{(ij)k}_{\beta\alpha} 
\, \boxpic{0.8}{figures/lambda.pdf}{
	\put (27,-19) {$k$}
	\put (-3,105) {$j$}
	\put (55,105) {$i$}
	\put (10,45) {$\beta$}~
}
\quad , \quad
\boxpic{0.8}{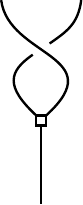}{
	\put (17,-13) {$k$}
	\put (-3,103) {$j$}
	\put (40,103) {$i$}
	\put (5,37) {$\alpha$}
}
	 = 
\sum_{\beta}
R^{-\,(ij)k}_{\beta\alpha} 
\, \boxpic{0.8}{figures/lambda.pdf}{
	\put (27,-19) {$k$}
	\put (-3,105) {$j$}
	\put (55,105) {$i$}
	\put (10,45) {$\beta$}~
}
\vspace{1ex}
\ee 
Another useful transformation we use is the $B$-matrix, which is defined by 
\be\label{eq:B-matrix-def}
(\id{i} \otimes c_{k,j}) \circ 
(\alpha \otimes \id{j})\circ \beta = \sum_{n,\gamma,\delta}
B^{(ijk)l}_{\gamma n \delta,\,\alpha m \beta} \, (\gamma \otimes \id{k})\circ \delta
~.
\ee
Graphically, the $B$-matrix is defined as follows: 
\vspace{1ex}
\be 
\boxpic{0.8}{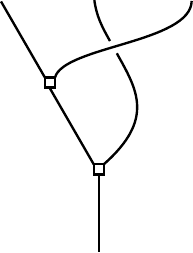}{
	\put (74,102) {$k$}
	\put (-3,102) {$i$}
	\put (34,102) {$j$}
	\put (8,60) {$\alpha$}
	\put (28,26) {$\beta$}
	\put (37,-12) {$l$}
	\put (18,42) {$m$}
} = \sum_{n,\gamma,\delta}
B^{(ijk)l}_{\gamma n \delta,\,\alpha m \beta} \, \boxpic{0.8}{figures/Fmatrix1.pdf}{
	\put (74,102) {$k$}
	\put (-3,102) {$i$}
	\put (35,102) {$j$}
	\put (8,60) {$\gamma$}
	\put (28,26) {$\delta$}
	\put (37,-12) {$l$}
	\put (18,42) {$n$}
}
\vspace{1ex}
\ee 
Using the $F$- and $R$-matrices, one can write for the expression on the left: 
\begin{align}\label{eq:B-matrix-computation}
(\id{i}\otimes c_{k,j})\circ (\alpha \otimes \id{j})\circ \beta 
&\overset{\eqref{eq:F-matrix}}= 
\sum_{p,\mu,\nu}{F^{(ikj)l}_{\mu p\nu, \alpha m \beta} (\id{i}\otimes c_{k,j})\circ (\id{i}\otimes \nu)\circ \mu}
\nonumber\\
&\overset{\eqref{eq:R-matrix}}=
\sum_{p,\mu,\nu,\lambda}{F^{(ikj)l}_{\mu p\nu, \alpha m \beta} R^{(kj)p}_{\lambda \nu}(\id{i}\otimes \lambda)\circ \mu}
\nonumber\\
&\overset{\eqref{eq:G-matrix}}=
\sum_{n,\gamma,\delta}\sum_{p,\mu,\nu,\lambda}{F^{(ikj)l}_{\mu p\nu, \alpha m \beta} R^{(kj)p}_{\lambda \nu} G^{(ijk)l}_{\gamma n \delta, \mu p \lambda}(\gamma \otimes \id{k})\circ \delta}
\end{align}
Inserting this into \eqref{eq:B-matrix-def} and comparing both sides we obtain an expression for the $B$-matrix in terms of $F$- and $R$-matrices
\be\label{eq:B-matrix-explicit} 
B^{(ijk)l}_{\gamma n \delta, \alpha m \beta} = \sum_{p,\mu,\nu, \lambda}{F^{(ikj)l}_{\mu p \nu,\alpha m \beta} R^{(kj)p}_{\lambda\nu} G^{(ijk)l}_{\gamma n \delta,\mu p \lambda}}~.
\ee
The reverse $B$-matrix is defined by 
\be
(\id{i} \otimes c_{j,k}^{-1}) \circ 
(\alpha \otimes \id{j})\circ \beta = \sum_{n,\gamma,\delta}
B^{-\,(ijk)l}_{\gamma n \delta,\,\alpha m \beta} \, (\gamma \otimes \id{k})\circ \delta
\ee
and pictorially:
\vspace{1ex}
\be 
\boxpic{0.8}{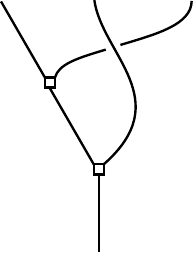}{
	\put (74,102) {$k$}
	\put (-3,102) {$i$}
	\put (34,102) {$j$}
	\put (8,60) {$\alpha$}
	\put (28,26) {$\beta$}
	\put (37,-12) {$l$}
	\put (18,42) {$m$}
} = \sum_{n,\gamma,\delta}
B^{-(ijk)l}_{\gamma n \delta,\,\alpha m \beta} \, \boxpic{0.8}{figures/Fmatrix1.pdf}{
	\put (74,102) {$k$}
	\put (-3,102) {$i$}
	\put (35,102) {$j$}
	\put (8,60) {$\gamma$}
	\put (28,26) {$\delta$}
	\put (37,-12) {$l$}
	\put (18,42) {$n$}
}
\vspace{1ex}
\ee 
Its computation resembles that of the $B$-matrix according to \eqref{eq:B-matrix-computation} but using in the second step the reversed $R$-matrix from \eqref{eq:R-matrix} because we work with the reverse braiding. One can easily verify that 
\be\label{eq:B-matrix-rev-explicit}
B^{-\,(ijk)l}_{\gamma n \delta, \alpha m \beta} = \sum_{p,\mu,\nu, \lambda}{F^{(ikj)l}_{\mu p \nu,\alpha m \beta} R^{-\,(kj)p}_{\lambda\nu} G^{(ijk)l}_{\gamma n \delta,\mu p \lambda}}~.
\ee

\newcommand{\arxiv}[2]{\href{http://arXiv.org/abs/#1}{#2}}
\newcommand{\doi}[2]{\href{http://doi.org/#1}{#2}}

\end{document}